\documentclass[letterpaper,twocolumn,10pt,
review,
annonymous,
]{article}
\usepackage{usenix2020-09}
\usepackage{amsfonts}
\usepackage{bbm}
\usepackage{booktabs}
\usepackage{makecell}
\usepackage{tikz}
\usetikzlibrary{arrows.meta}
\usepackage{adjustbox}
\usepackage{enumitem}
\usepackage{wrapfig}
\usepackage{subcaption}
\usepackage{xspace}
\usepackage[svgnames]{xcolor}

\usepackage{enumitem}

\usepackage[switch]{lineno}

\usepackage{algorithm}
\usepackage[noend]{algpseudocode}
\usepackage{mathrsfs}
\usepackage{amsmath}
\usepackage{amssymb}
\usepackage[makeroom]{cancel}
\usepackage{amsthm}

\usepackage{booktabs}

\usepackage{colortbl}

\usepackage{array}      
\usepackage{booktabs}   
\usepackage{breqn}
\usepackage{tabularx}

\newtheorem{thm}{Theorem}
\newtheorem{lem}[thm]{Lemma}
\newtheorem{cor}[thm]{Corollary}

\usepackage[framemethod=tikz]{mdframed}
\newmdenv[
    linecolor=gray,
    linewidth=1pt,
    topline=false,
    bottomline=false,
    rightline=false,
    skipabove=2pt,
    skipbelow=1pt,
    leftmargin=0,
    rightmargin=10pt,
    innerleftmargin=3pt,
    innerrightmargin=0pt,
    innertopmargin=1pt,
    innerbottommargin=2pt,
    backgroundcolor=white
]{txtframe}
\newmdenv[
    linecolor=gray,
    linewidth=1pt,
    roundcorner=4pt,
    skipabove=2pt,
    skipbelow=1pt,
    leftmargin=0,
    rightmargin=0,
    innerleftmargin=6pt,
    innerrightmargin=6pt,
    innertopmargin=5pt,
    innerbottommargin=5pt,
    backgroundcolor=gray!10
]{graytxtframe}
\newmdenv[
    linecolor=gray,
    linewidth=1pt,
    topline=true,
    bottomline=true,
    leftline=false,
    rightline=false,
    skipabove=1pt,
    skipbelow=1pt,
    leftmargin=10pt,
    rightmargin=10pt,
    innerleftmargin=2pt,
    innerrightmargin=2pt,
    innertopmargin=2pt,
    innerbottommargin=2pt,
    backgroundcolor=white
]{insightframe}

\usepackage[utf8]{inputenc}
\usepackage[T1]{fontenc}

\newif\ifshowcomment
\showcommenttrue

\newcommand{\sys}{{\textsc{ReGuard}}\xspace}

\newcommand{\tocite}[1]{{\textcolor{red}{\textbf{[~]}}}}
\newcommand{\toref}[1]{\textcolor{red}{\textbf{N}}}

\newcommand{\eg}{\emph{e.g.,} }
\makeatletter
\newcommand{\etc}{etc\@ifnextchar.{}{.\@\xspace}}
\makeatother

\newcommand{\hhy}[1]{{\footnotesize\color{orange}[hongyu: #1]}}
\newcommand{\maria}[1]{{\footnotesize\color{blue}[maria: #1]}}

\ifshowcomment
\newcommand{\TODO}[1]{\textcolor{red}{{[\small\textsf{{TODO: #1}}}]}}
\newcommand{\NOTE}[1]{\textcolor{orange}{{[\small\textsf{{NOTE: #1}}}]}}
\else
\newcommand{\TODO}[1]{}
\newcommand{\NOTE}[1]{}
\fi
\newcommand{\x}{$\times$\xspace}

 \newcommand{\myitem}[1]{\vspace*{0.02in}\noindent\textbf{#1}}
\newcommand{\remove}[1]{}

\newcommand{\mypar}[1]{{\noindent\bf #1.\ }}

 \usepackage{tcolorbox}
\tcbuselibrary{listings}
\usepackage{graphicx}

\newtcolorbox{promptbox}{
  colback=gray!5,
  colframe=gray!60,
  listing only,
  listing options={
    basicstyle=\ttfamily\small,
    breaklines=true,
    columns=fullflexible
  }
}

\usepackage[numbers,square,sort&compress]{natbib}
\usepackage{tikz}
\usepackage{amsmath}
\usepackage{textcomp}

\graphicspath{{assets/badges/}}
\usepackage[available,functional,reproduced]{usenixbadges}

\begin{document}

\date{}


\title{Worst-Case Discovery and Runtime Protection for RL-Based Network Controllers}

\author{
{\rm Hongyu H\`e}\\
Princeton University
\and
{\rm Minhao Jin}\\
Princeton University
\and
{\rm Maria Apostolaki}\\
Princeton University
}

\maketitle
\pagestyle{empty}
\thispagestyle{empty}

\begin{abstract}
RL-based controllers achieve strong average-case performance in networking tasks such as congestion control and adaptive bitrate streaming.
Yet their performance can degrade severely under network conditions where strong performance is still achievable.
Identifying such conditions and quantifying the resulting performance gap is intractable by enumeration, while the sequential and closed-loop nature of RL controllers makes formal verification methods impractical.

We present \sys, a framework that discovers worst-case scenarios for a given RL controller and protects it against them at inference time without retraining.
\sys formulates scenario discovery as a bilevel regret-maximization problem, certifying a lower bound on the controller's worst-case performance gap.
\sys then analyzes the discovered counterfactual trajectories and compiles them into lightweight logic rules that intervene only when a risky state is detected, leaving the controller's behavior unchanged otherwise.

We evaluate \sys across three RL-based network controllers: Pensieve, Sage, and Park.
\sys discovers scenarios in which each controller suffers a 43--64\% performance loss relative to what is achievable under the same conditions.
\sys exposes performance gaps 57\% to 6\x larger than those found by the strongest baselines and shrinks those gaps by 79--85\% in its first iteration via lightweight rule-based protection, iteratively driving the remaining gap toward nominal levels.
\sys's protection extends beyond the scenarios it discovers, improving performance across a wide range of network conditions.
\end{abstract}

\section{Introduction} \label{sec:intro}

For more than a decade, RL-based and self-adapting controllers have been proposed for networked and cloud systems, including adaptive bitrate streaming, congestion control, and resource management~\cite{yan2020insitu,mao_neural_2017_pensieve,xia2022genet,abbasloo2020orca,hosseinzadeh2023cane,he2024peros}.
On standard benchmarks, these learned controllers consistently outperform handcrafted heuristics, making a strong case for production deployment.
Yet before an operator can justify deploying a learned controller in production, she needs answers to three questions that average-case benchmarks cannot provide.
First, are there realistic network scenarios, meaning sequences of network conditions, under which the RL controller performs far worse than what is achievable?
Second, when such scenarios arise, how much achievable performance does the controller forfeit?
Third, can we mitigate these worst-case failures at runtime without sacrificing the controller's strong performance under normal operating conditions?

Answering these three questions is difficult because the controller and the network interact in a closed loop over time.
An RL controller is a black-box neural policy whose current action alters future network states, such as queue occupancy and packet loss, which in turn shape the observations the controller receives at subsequent steps.
Because the space of multi-step network scenarios grows exponentially with the control horizon, existing tools struggle to either find true worst-case failures or fix them safely.
Formal verifiers and symbolic analyzers~\cite{namyar2024metaopt,eliyahu2021whirl,raghunathan2025quasi,he2026typonet} require an explicit mathematical model of both the policy and the packet-level environment dynamics, which is rarely available for production systems, while scripted testing tools such as packetdrill~\cite{cardwell2013packetdrill} check only fixed packet traces.
Approximating an RL controller with a decision tree or surrogate model~\cite{jacobs2022trustee,meng_interpreting_2020_metis} loses fidelity in the tail regimes where worst-case failures hide.
Conversely, black-box fuzzers such as CC-Fuzz~\cite{ray2022ccfuzz}, AdvNet~\cite{ahmed2026advnet}, and Gilad et al.~\cite{gilad_robustifying_2019} search for conditions that minimize raw controller reward, which often surfaces pathological network outages where no policy can succeed~\cite{arun2021toward}.
On the mitigation side, curriculum-learning and fine-tuning frameworks such as Genet~\cite{xia2022genet} and Mowgli~\cite{agarwal2025mowgli} update the neural policy's weights globally, which frequently degrades performance on normal traffic.
Meanwhile, point solutions such as Canopy~\cite{yang2026canopy} and Galileo~\cite{saxena2026galileo} rely on domain-specific input-output constraints for congestion control or queueing formulas for RPC scheduling that do not generalize across networking tasks.

We present \sys, a general framework that discovers avoidable worst-case failures and protects black-box RL controllers at runtime without requiring a formal system model or policy retraining~\cite{gong2023fmml,he2025lejit,he2026netnomos,he2026autogram}.
Our first key insight is that a genuine controller failure is defined not by low absolute reward, but by high \emph{regret}, meaning a large gap between what the RL controller achieves and the best performance achievable under the exact same network conditions.
Measuring regret separates avoidable controller mistakes from harsh network conditions where every algorithm fails.
However, computing the exact optimal policy for every candidate scenario is intractable in closed-loop environments.
To make discovery tractable, \sys formulates scenario search as a bilevel regret-maximization problem over feasible network conditions.
The outer loop uses an RL-based solver to search for the scenario that maximizes the target controller's regret, while the inner loop approximates the scenario-specific optimum using the best-performing policy from a portfolio of established networking heuristics.
Because decades of networking research have produced heuristics that excel in complementary operating regimes, their per-scenario upper envelope closely approximates the true optimum and yields a certified lower bound on the controller's worst-case regret.

Our second key insight is that the same heuristic portfolio that certifies a worst-case failure also provides the counterfactual supervision needed to fix it at runtime.
Whenever \sys discovers a high-regret scenario, it compares the failing RL controller's trajectory against the trajectory of the heuristic that succeeded under those same conditions.
This comparison reveals which observable controller states are risky and whether the controller acted too aggressively or too conservatively at each step.
\sys distills these state-level counterfactuals into compact, threshold-based logic rules~\cite{he2026netnomos,he2025lejit} that evaluate only the controller's current observation vector at inference time.
When no rule matches, \sys abstains and leaves the RL controller's action untouched, preserving its nominal performance; when a risky state is detected, \sys applies a bounded correction in the direction validated by the counterfactual heuristic.

We evaluate \sys across three state-of-the-art RL controllers spanning distinct networking domains: Pensieve for adaptive bitrate streaming, Sage for congestion control, and Park for load balancing.
\sys discovers feasible scenarios in which these controllers suffer a 43--64\% performance loss relative to the reference portfolio, exposing gaps up to 7\x larger than those found by Genet~\cite{xia2022genet}, Indago~\cite{biagiola2024indago}, and Gilad et al.~\cite{gilad_robustifying_2019}.
At runtime, \sys's learned rules shrink these worst-case gaps by 79--85\% in a single iteration while preserving or slightly improving performance on normal traces, and two additional refinement rounds drive the remaining gap close to nominal levels.

Critically, \sys's protection generalizes beyond the scenarios used to learn the rules, confirming that it repairs underlying policy flaws rather than memorizing specific traces.
Controllers protected by \sys outperform their unprotected counterparts on challenging scenarios generated by Genet, Indago, and Gilad et al., even though those scenarios were never seen during rule synthesis.
Re-running \sys's bilevel search against the protected controller likewise yields much smaller worst-case regret, showing that \sys shrinks the controller's overall vulnerability surface.
Finally, because \sys evaluates its rules through a compact prefix trie over raw state features, it fits well within each controller's real-time decision budget and produces human-readable rules that expose concrete policy bugs, such as Pensieve's tendency to select unsustainable bitrates during bandwidth drops and Park's bias toward overloaded slow servers.




\remove{
\paragraph{Contributions.}
This paper makes three contributions.
\begin{enumerate}[leftmargin=*, itemsep=0pt, topsep=2pt]
 \item We formulate performance-failure discovery for RL-based network control as a bilevel solvability-weighted regret-maximization problem with appealing theoretical guarantees. 
 We show that the discovered network scenario certifies a lower bound on the worst-case solvability-weighted performance-failure objective.
 \item We develop practical implementations of this formulation across three RL-based network control tasks, certifying performance gaps of at least 80\% for Pensieve, at least 30\% for Sage, and at least X\% on a third use case.
 \item We develop two-stage rule-based run-time protection that protects the controller without fine-tuning and reveals the discovered failures in an operator-meaningful form.
 The first stage identifies risky controller states, and the second stage prescribes targeted action adjustments based on the reference policy's behavior.
 Across the three use cases, this protection reduces the discovered performance gaps by over X\% while preserving nominal performance, unlike naive fine-tuning.
\end{enumerate}
}

\section{Motivation}

\begin{figure}[t]
    \centering
    \begin{subfigure}[b]{0.35\linewidth}
        \centering
        \includegraphics[width=\linewidth]{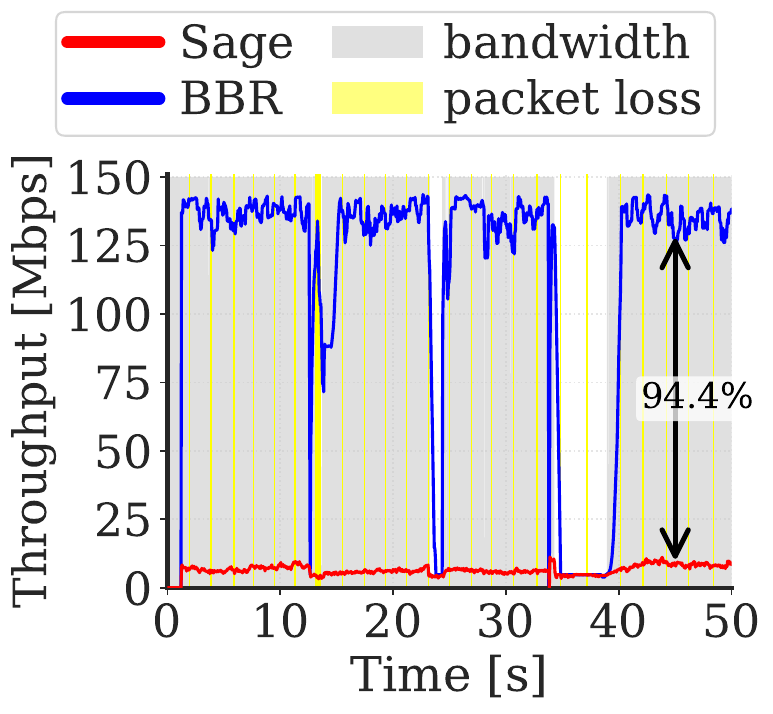}
        \caption{Vanilla Sage}
        \label{fig:motivation_timeseries}
    \end{subfigure}
    \hfill
    \begin{subfigure}[b]{0.64\linewidth}
        \centering
        \includegraphics[width=\linewidth]{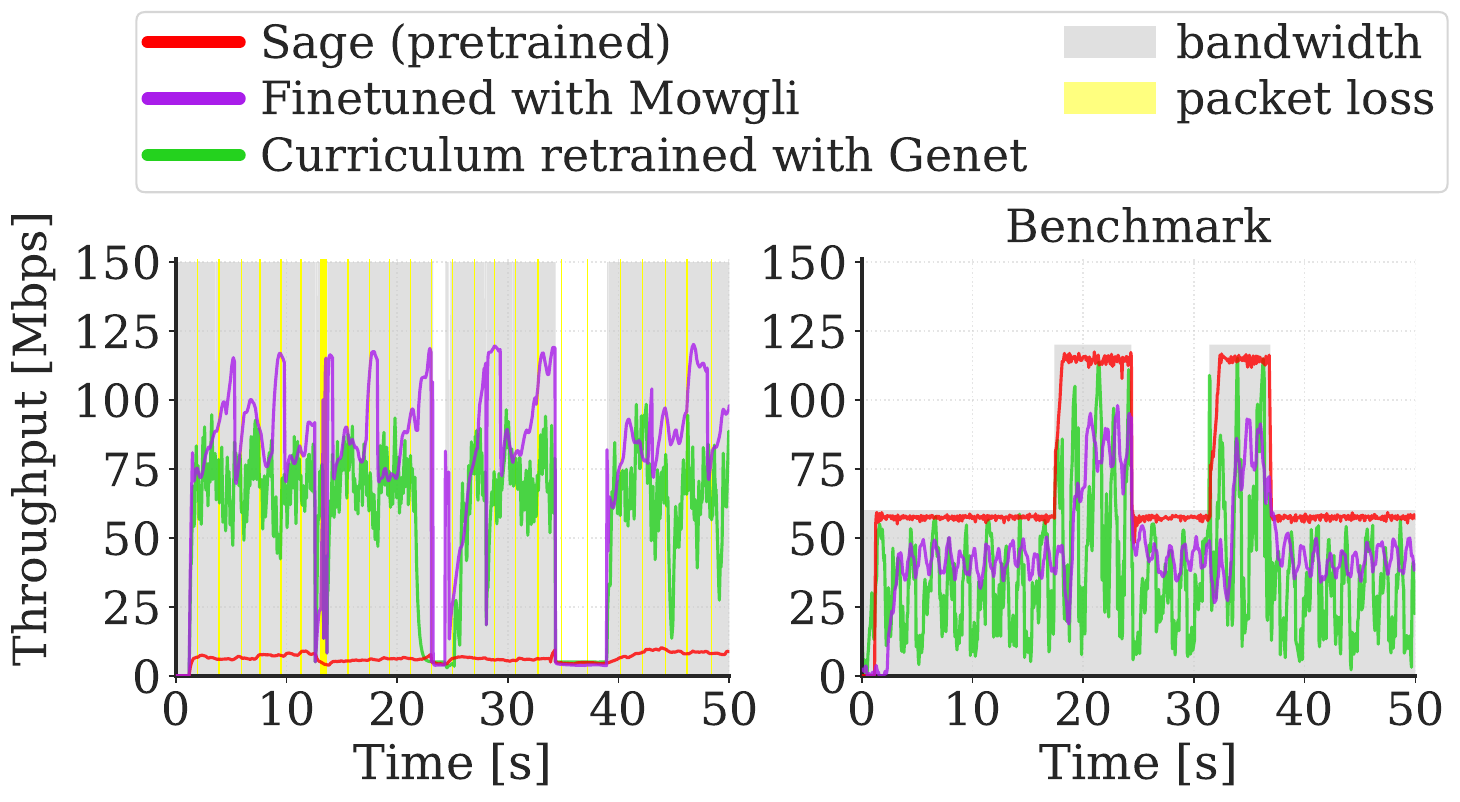}
        \caption{Fine-tuned Sage}
        \label{fig:motivation_improved}
    \end{subfigure}
    \caption{Sage suffers severe, avoidable throughput collapse under specific network conditions, and fine-tuning fails to fix the gap without hurting normal performance.
    (a) Under a feasible lossy scenario, vanilla Sage~\cite{yen2023sage} achieves 94\% lower throughput than BBR.
    (b) Fine-tuning Sage on challenging scenarios neither restores achievable throughput in that scenario nor preserves Sage's original throughput on normal benchmarks.}
    \label{fig:motivation}
\end{figure}


This section uses RL-based congestion control as an example use case to explain why strong benchmark performance is insufficient for real-world RL deployment in networked systems and to derive the requirements for \sys.
We detail the use cases considered in Table~\ref{tab:maxregret_impl} and Appendix~\ref{app:use-cases}.

\subsection{Use Case: RL-Based Congestion Control}

Consider an operator managing a private datacenter network.
To improve performance for latency-sensitive and high-throughput workloads such as distributed machine-learning training~\cite{he2024gcow,zhou2025proxy}, she considers replacing the currently deployed BBR controller with Sage~\cite{yen2023sage}, a state-of-the-art RL-based congestion controller that achieves higher average throughput and lower latency.
She evaluates Sage on representative internal traces, and the average results are encouraging.
However, strong average performance does not rule out severe tail failures in production.
As Fig.~\ref{fig:motivation_timeseries} illustrates, under a feasible scenario where high bottleneck bandwidth is paired with intermittent packet loss, BBR sustains near-capacity throughput whereas Sage overreacts to the loss events, backs off sharply, and forfeits 94\% of the achievable throughput.
Before deploying Sage, the operator therefore needs to answer a fundamental question: \textbf{Under what feasible network scenarios does Sage fail, and how much achievable performance can it lose?} 

Her first instinct is formal verification, which can certify network properties when an explicit mathematical model of the system is available~\cite{namyar2024metaopt,eliyahu2021whirl,raghunathan2025quasi,he2026typonet}.
For instance, QuASI answers queue-related queries from coarse packet counts~\cite{raghunathan2025quasi}, MetaOpt and whiRL verify policies within simplified analytical models~\cite{namyar2024metaopt,eliyahu2021whirl}, and TypoNet translates static network configurations into symbolic logic~\cite{he2026typonet}.
However, verifying Sage's closed-loop throughput in emulation or production would require formally modeling its multi-layer neural network alongside Linux kernel pacing, ACK clocking, queue evolution, and time-varying link conditions over hundreds of RTTs, which is intractable in practice.

Her next thought is adversarial machine learning, which stress-tests neural networks using gradient-based input perturbations such as projected gradient descent (PGD)~\cite{goldstein1964convex,levitin1966constrained,madry2018towards,boyd2004convex}.
However, PGD perturbs Sage's 69-dimensional observation vector directly without respecting the physical laws that couple those features (Table~\ref{tab:sage_state_features}).
For example, PGD can simultaneously increase Sage's observed delivery-rate ratio and its packet-loss ratio to values that imply full link utilization alongside massive queue drops on an uncongested link, a state that no real bottleneck queue can ever produce.
While static classifiers and offline trace generators can enforce snapshot feature constraints or data invariants~\cite{jin2024pants,ben2024cafa,pierazzi2020problemspace,gong2024zoom2net,he2026netnomos,he2026autogram,jin2027tracebleed}, an RL controller operates in a closed loop where each action changes the future network state and subsequent observations.
Ensuring that an adversarial perturbation corresponds to a physically realizable trajectory requires executing or modeling the entire closed-loop environment, bringing the operator back to the same barrier that blocked formal verification.

Her remaining option for discovery is random or heuristic fuzzing.
Because a scenario unfolds over $T$ control intervals with multiple link parameters per step, the space of feasible trajectories grows exponentially with $T$.
Random sampling spends nearly all of its budget on benign scenarios where every policy succeeds or on degenerate outages where every policy fails, offering no guarantee of finding selective blind spots like Fig.~\ref{fig:motivation_timeseries} and no bound on how close the worst observed gap is to the true worst case.

Even if the operator discovers a severe failure scenario like Fig.~\ref{fig:motivation_timeseries}, she still needs a way to fix Sage without breaking its normal behavior.
The standard remedy in machine learning is fine-tuning or curriculum retraining on the failure scenarios~\cite{kumar2019stabilizing,thompson-etal-2019-overcoming,kumar2022finetuning,luo2024optimistic,tsipras2019robustness}.
Because neural network weights are shared across all inputs, however, gradient updates that push Sage to send more aggressively during lossy scenarios globally shift its policy across benign regimes as well.
Fig.~\ref{fig:motivation_improved} demonstrates this trade-off: fine-tuning Sage on challenging scenarios only partially recovers throughput in the failure scenario of Fig.~\ref{fig:motivation_timeseries} while degrading its throughput on normal benchmarks where the original model excelled.
The operator is left with an unacceptable dilemma: \textbf{Should she optimize the controller for normal-case performance or for worst-case robustness?}

\subsection{Requirements \& Design Principles}
The congestion-control example highlights the need for a framework that complements pretrained RL controllers to improve their robustness and deployability.

First, the framework must discover feasible network scenarios under which the target controller performs substantially worse than what is achievable under the same conditions.
This gap captures the true worst case for the controller, isolating settings where the controller fails to adapt even though strong performance remains possible.
For instance, the scenario in Fig.~\ref{fig:motivation_timeseries}, where Sage achieves 94\% lower throughput than BBR, represents an avoidable failure that an operator must identify and mitigate before deployment.

Second, the protection mechanism must preserve the controller's strong nominal performance.
Any mitigation that degrades normal-case performance defeats the purpose of deploying an RL controller, as illustrated by the fine-tuning degradation in Fig.~\ref{fig:motivation_improved}.

Finally, the framework must remain practical to deploy without requiring a formal model of the controller, the environment, or their closed-loop dynamics in advance.

\section{Overview} \label{sec:overview}

\begin{figure*}[t]
    \centering
    \begin{adjustbox}{width=1.0\textwidth,center}
        \includegraphics{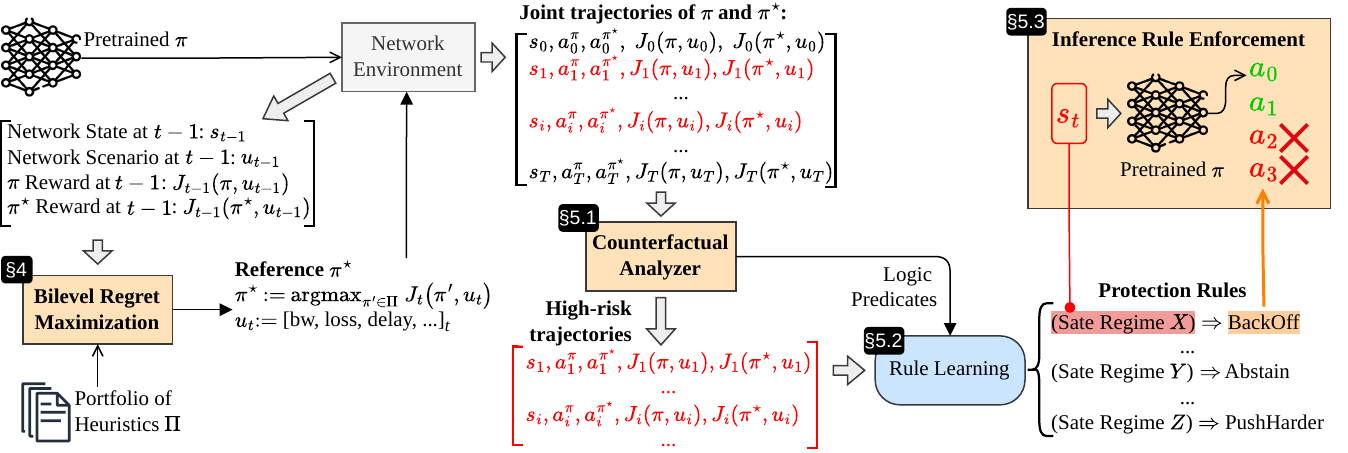}
    \end{adjustbox}
    \caption{
    \sys discovers scenarios that maximize the targeted RL controller's regret by directly interacting with the network environment in which the controller was trained and by using a portfolio of heuristics to approximate the scenario-specific optimum.
    It then analyzes the resulting counterfactuals to extract recurring risky patterns and corrective directions.
    Finally, \sys compiles those counterfactuals into logic rules that override the controller's decision at inference time only when intervention is strictly necessary.}
    \label{fig:overview}
\end{figure*}

Following these design principles, \sys takes the form shown in Fig.~\ref{fig:overview}.
Next, we highlight the main insights that drove its design.

\myitem{By framing discovery as regret maximization over a bilevel program, we can provide a tightness guarantee on how close the discovered failure is to the true worst case.}
To find scenarios that leave the most achievable performance on the table, we guide the search toward worst-case regret, defined as the performance gap between what the controller achieves and what is achievable under the same conditions.
\sys frames discovery as an optimization over all feasible network conditions for the scenario that maximizes regret.
Because \sys optimizes regret directly rather than sampling randomly, the returned scenario carries a formal bound (Theorem~\ref{thm:exact}): the gap between the regret it achieves and the true worst-case regret is bounded by the suboptimality of the reference policy.
Computing the exact optimal policy is intractable for a stochastic dynamical system over long horizons, yet decades of networking research have produced heuristics that each excel in different operating regimes.
The best-performing heuristic on each scenario forms an upper envelope that serves as a strong proxy for the scenario-specific optimum.
\sys therefore uses a bilevel structure in which the inner problem selects the best reference heuristic for each candidate scenario that the outer problem explores.

\myitem{The same upper envelope that certifies the failure also provides the counterfactual signal needed for protection.}
The portfolio of heuristics serves a second role beyond certification.
When \sys discovers a high-regret scenario, it identifies both which heuristic succeeded and what that heuristic did differently at each decision point.
This comparison provides a concrete counterfactual: under the same network conditions, the reference heuristic backed off while the RL controller pushed harder, or vice versa.
By comparing the controller and reference trajectories, \sys pinpoints risky states and determines the corrective direction needed to prevent performance degradation.
To cover distinct blind spots in the controller, \sys searches for a diverse set of high-regret scenarios rather than relying on a single trajectory.

\myitem{The highest-regret scenarios cluster around a few recurring and high-risk decision-making flaws, making surgical intervention both possible and effective.}
The highest-regret scenarios reflect systematic flaws in the controller's policy rather than isolated anomalies.
Because the search objective rewards the same controller weakness whenever it surfaces, many discovered scenarios expose the same underlying mistake.
For instance, \sys scenarios reveal that Park repeatedly routes small jobs to slow servers even when fast servers sit idle.
Because the discovered scenarios cluster around a small number of recurring state patterns, \sys can distill them into a compact set of logic rules over interpretable predicates that cover a wide range of failure conditions while remaining auditable by operators.


\remove{
\section{Overview of \sys} \label{sec:overview}

\mypar{Solvability-weighted regret maximization targets worst-case failures where good performance is achievable} \maria{we have not explained what regret is here}
\hhy{thanks for pointing it out.}
\sys formulates the discovery of challenging network scenarios for an RL controller as a bilevel solvability-weighted regret maximization problem (\S\ref{subsec:problem}).
The key idea is that not every low-reward scenario is informative.
If both the RL controller and a strong reference policy perform poorly, then the scenario may simply be intrinsically challenging. \maria{this is repetitive}
\sys therefore searches for scenarios in which the controller performs poorly relative to what is still achievable under the same network conditions.
Its objective prefers large performance gaps in scenarios where the consideration policy still attains meaningful reward, thereby focusing the search on substantial performance failures rather than arbitrary degradation.
This framing yields more than a better search objective.
It also yields a sharper warranty: the closer the reference policy is to the true optimum, the closer the resulting scenarios are to the true worst-case failures of the RL controller.
This insight is apparent in Fig.~\ref{fig:gap_plots}, where \sys finds the largest performance gaps in all three applications, reaching a mean log performance gap of 4.84 for Pensieve~\cite{mao_neural_2017_pensieve} and mean performance gaps of 63.84\% for Sage~\cite{yen2023sage} and 43.49\% for Park~\cite{mao2019park}. \maria{Are we saying here that by optimizing for the gap we get the largest gap?}
\hhy{Good point, indeed. We are not just maximizing the gap. We are maximizing a weighted gap.}
\maria{maybe what we want to say here is that if we were only comparing with one heuristic, we risk ignoring bad cases where the heuristic fails, while if we optimize against the optimal, we are heuristic-agnostic. Still at this point the solvability weighted is not explained in that there is no problem where the solution is correct and we also do not explicitely state how we define solvability  }

\mypar{The unattainable optimal policy can be well approximated by the upper envelope of simple heuristics.} 
In practice, the exact optimal policy for each network scenario is usually unattainable. \maria{is the problem that it is unattainable or is it unknown?}
\sys bypasses this obstacle by using a portfolio of simple heuristics as a practical reference.
Although any single heuristic may be weak or regime-specific, different heuristics perform best in different operating regimes~\cite{pappone2025mutant,yen2023sage,agarwal2025mowgli}.
Their collective upper performance envelope therefore provides a strong approximation to the scenario-specific optimum.
This approximation is central to \sys.
A stronger reference sharpens the outer search \maria{the outer search is not defined} toward real avoidable mistakes, reduces search error in the bilevel objective,\maria{why or add citation?} and strengthens the resulting certificate on the worst-case performance gap.
It also provides the counterfactual supervision needed for protection: under the same scenario, the reference policy reveals whether the controller's decision was too aggressive or too conservative, and thus what corrective direction \sys should take to adjust the suboptimal actions of the RL controller.
Fig.~\ref{fig:heatmaps} illustrates this insight: the resulting \sys configuration removes 85.04\% of the normalized performance gap for Pensieve, 80.01\% for Sage, and 79.12\% for Park on the most challenging scenario family in each application, while still transferring to other scenario families.
The reference policy is therefore not merely a baseline for evaluation, but the mechanism that bridges worst-case discovery and targeted failure protection in practice.

\mypar{Learning from the most challenging failures enables surgical protection that generalizes to weaker ones} \maria{this myitem repeats the last sentence of the previous one. the emphasis here needs to be on the rules vs retraining not again on the "hardest" failures}
\sys learns its protection rules directly from the highest-gap scenarios surfaced by the bilevel search.
This is the crucial design choice behind its protection.
The most challenging failures are valuable not only because they are severe, but also because they concentrate the controller's most costly and avoidable mistakes into a much cleaner corrective signal than diffuse sampling from mildly challenging environments.
Fine-tuning and retraining apply global updates to fix what is often a local defect, which risks degrading the controller in regimes where it already works well.
In contrast, \sys extracts simple logic rules that identify risky state regimes and apply bounded action corrections only when intervention is warranted.
Because \sys is derived from the most substantial performance failures, it naturally protects against weaker failures that arise from the same underlying decision-making flaws.
This approach yields a cascading benefit: by surgically protecting the controller against its worst-case failures, \sys also improves robustness more broadly while preserving nominal performance in low-risk regimes where the base controller already works well.
Fig.~\ref{fig:ablation} backs up this insight directly: \sys consistently removes more of the original performance gap when it is derived from more challenging source scenarios than when it is derived from weaker scenarios.
On the most challenging scenario family in each application, the resulting \sys configuration removes 85.04\% of the normalized performance gap for Pensieve, 80.01\% for Sage, and 79.12\% for Park.
\maria{overview is not a summary of eval we should not point to that many results but it should point to the main design points. Right now we are overly focusing on the large paps. for example, the logic analyzer of the overview figure is not even mentioned.}
\hhy{The logic analyzer is about counterfactual analysis for creating \sys's protection rules. I thought you wanted me to tune it down? I've sneaked it into the 2nd point.}

}

\section{Bilevel Regret Maximization over Stochastic Dynamical Systems}

\subsection{Problem Formulation}
\label{subsec:problem}

Evaluating whether an RL controller exhibits a severe performance failure requires an informative frame of reference: if an environment is so degraded that no algorithm succeeds, poor performance reflects an unavoidable barrier rather than a controller flaw.
A genuine failure occurs only when the network allows high performance, yet the learned controller forfeits it.
\sys formalizes this principle by searching for a feasible scenario in $\mathcal{E}$ that maximizes the regret of a fixed pretrained controller $\pi$ relative to the best reference policy in class $\Pi$ under that same scenario.

\mypar{Scenario, reward, and regret}
A scenario is a finite resource sequence $e=(u_0,\dots,u_{T-1}) \in \mathcal{E}$, where each $u_t \in \mathcal{U}$ specifies the exogenous network resources or workload inputs available at step $t$.
Thus, $e$ represents a time-varying network trajectory rather than a static snapshot.
For any control policy $\pi'$, its expected cumulative reward in scenario $e$, the scenario-specific optimal reference policy $\pi_e^\star$ within reference class $\Pi$, and the resulting regret $R(e)$ of the pretrained controller $\pi$ are defined as
\begin{equation}
\label{eq:scenario_reward_regret}
\begin{aligned}
J(\pi';e) &:= \mathbb{E}\!\left[\sum_{t=0}^{T-1} \gamma^t r_t \;\middle|\; \pi',e\right], \\
\pi_e^\star &\in \arg\max_{\pi' \in \Pi} J(\pi';e), \\
R(e) &:= J(\pi_e^\star;e) - J(\pi;e).
\end{aligned}
\end{equation}

Measuring regret $R(e)$ isolates genuine controller failures.
If a scenario is intrinsically bad for every policy, then both $J(\pi_e^\star;e)$ and $J(\pi;e)$ are low and the regret stays small.
A large regret instead means that the controller leaves substantial reward on the table in a feasible scenario where much better performance is still achievable.

\mypar{Feasible search space}
Operators define $\mathcal{E}$ to capture realistic deployment constraints and exclude unrepresentative conditions:
\begin{equation}
\label{eq:efeas}
\mathcal{E}
:=
\{ e : e \models c_j,\; j=1,\dots,m \}.
\end{equation}
Constraints $c_j$ encode operational bounds (\eg bounding packet loss $0 \le \mathrm{loss}_t \le 0.01\%$, capping maximum link bandwidth~\cite{he2024gcow}, limiting rates of capacity change, or restricting job arrival envelopes~\cite{ustiugov2023invitro}), concentrating discovery on realistic operating regimes rather than pathological artifacts.

\mypar{Bilevel regret maximization}
\sys targets the following bilevel problem:
\begin{equation}
\label{eq:outer_true}
e^\star \in \arg\max_{e \in \mathcal{E}}
\left[ J(\pi_e^\star;e) - J(\pi;e) \right],
\end{equation}
or equivalently,
\begin{equation}
\label{eq:bilevel_true}
\begin{aligned}
\operatorname*{maximize}_{e \in \mathcal{E}}
\quad & J(\pi_e^\star;e) - J(\pi;e) \\
\text{subject to}\quad & \pi_e^\star \in \arg\max_{\pi' \in \Pi} J(\pi';e).
\end{aligned}
\end{equation}
The outer level searches over feasible scenarios, while the inner level picks the strongest reference policy for each candidate scenario.
A high-scoring scenario is therefore a concrete failure case in which the controller underperforms by a large and avoidable amount.

This formulation operates over full trajectories rather than static snapshots.
The outer variable $e$ is a time series of resources, and both $J(\pi;e)$ and $J(\pi_e^\star;e)$ depend on the closed-loop evolution over the horizon, where abrupt bandwidth shifts and delayed congestion feedback govern how controllers adapt~\cite{meng2026confucius,zhou2025proxy}.
Unlike static-input analyses such as MetaOpt~\cite{namyar2024metaopt}, \sys searches over time-varying network scenarios for dynamical controllers.

\mypar{Outer RL solver and practical reference}
We use RL only as an offline solver for the outer search over scenarios.
This outer solver is distinct from the control policies $\pi$ and $\pi_e^\star$, which act \emph{within} a fixed scenario.
RL fits the outer search naturally because the scenario itself is sequential and the search requires no explicit formal model of the pretrained controller or the network dynamics.

In practice, the exact inner solution may be unattainable, especially in continuous-action settings.
We therefore use an approximate reference policy $\hat{\pi}_e^\star \in \Pi$ and solve
\begin{equation}
\label{eq:outer_hat}
\hat{e} \in \arg\max_{e \in \mathcal{E}}
\left[ J(\hat{\pi}_e^\star;e) - J(\pi;e) \right].
\end{equation}
The theoretical target is Eqn.~\eqref{eq:bilevel_true}, while the implemented procedure approximately solves Eqn.~\eqref{eq:outer_hat}.

\subsection{Guarantees}
\label{sec:guarantees_sys}

The bilevel formulation provides a formal guarantee with a clear operational meaning: the scenario returned by \sys certifies a lower bound on the pretrained controller's worst-case avoidable performance gap over the feasible space $\mathcal{E}$.

For readability, define
\begin{equation}
\label{eq:regret_defs}
R(e) := J(\pi_e^\star;e) - J(\pi;e),
\qquad
\hat{R}(e) := J(\hat{\pi}_e^\star;e) - J(\pi;e).
\end{equation}
Here $R(e)$ is the exact regret objective and $\hat{R}(e)$ is its practical approximation.

\mypar{Main guarantee}
Under the assumptions stated in Appendix~\ref{app:proofs}, let $\tilde e$ be the scenario returned by the outer RL solver, let $e^\star$ maximize the exact objective in Eqn.~\eqref{eq:outer_true}, let $\varepsilon$ denote the outer-solver suboptimality for Eqn.~\eqref{eq:outer_hat}, and let $\delta(e)$ bound the inner-reference error in scenario $e$.
Then
\begin{equation}
\label{eq:true_gap_alt}
R(\tilde e)
\ge
\max_{e \in \mathcal{E}} R(e) - \varepsilon - \delta(e^\star),
\end{equation}
or equivalently, $\max_{e \in \mathcal{E}} R(e) - R(\tilde e) \le \varepsilon + \delta(e^\star)$.

Equation~\eqref{eq:true_gap_alt} shows that the exact regret achieved by the returned scenario lower-bounds the worst-case exact regret over the feasible search space.
When the outer RL solver is near-optimal and the practical reference is accurate, this lower bound is tight, making the returned scenario a near-tight certificate of controller weakness in the searched regime.

\mypar{Supporting results}
The proof relies on two lemmas proved in Appendix~\ref{app:proofs}.
First, the outer RL solver returns an $\varepsilon$-optimal scenario for the approximate objective:
\begin{equation}
\label{eq:guar1_alt}
\hat{R}(\tilde e) \ge \max_{e \in \mathcal{E}} \hat{R}(e) - \varepsilon.
\end{equation}
Second, the approximate regret is a pointwise lower bound on the exact regret, with gap controlled by the inner-reference error:
\begin{equation}
\label{eq:pointwise_gap}
0 \le R(e) - \hat{R}(e) \le \delta(e),
\qquad \forall e \in \mathcal{E}.
\end{equation}
Combining Eqns.~\eqref{eq:guar1_alt} and \eqref{eq:pointwise_gap} yields Theorem~\ref{thm:exact} below.

\begin{txtframe}
\begin{thm}\label{thm:exact}
Under \ref{as:nonempty}--\ref{as:oracle}, the scenario $\tilde e$ returned by \sys certifies a lower bound on the worst-case exact regret objective, and this lower bound is within $\varepsilon + \delta(e^\star)$ of the exact worst case (Eqn.~\ref{eq:true_gap_alt}).\footnotemark
\end{thm}
\end{txtframe}\footnotetext{Proof in Appendix~\ref{app:proof-exact}.}

\mypar{Scope and interpretation of the guarantee}
Theorem~\ref{thm:exact} certifies that \sys discovers at least the observed avoidable regret, approaching the true worst case as the outer solver and reference approximation improve.
The bound loosens only through two separable offline errors: the outer search error $\varepsilon$, which reflects solver convergence, and the inner reference error $\delta(e^\star)$, which reflects portfolio coverage.
Expanding $\Pi$ reduces $\delta(e^\star)$ by capturing complementary heuristic strengths across regimes~\cite{pappone2025mutant}, tightening the certificate at the cost of offline evaluation compute.
Crucially, this computation runs entirely offline; neither $\Pi$ nor the outer solver executes at inference time.

\subsection{Implementation}
\label{subsec:maxregret_impl}

\begin{table*}[t]
\centering
\begin{adjustbox}{width=\textwidth,center}
\footnotesize
\begin{tabularx}{\textwidth}{>{\raggedright\arraybackslash}p{0.13\textwidth}>{\raggedright\arraybackslash}p{0.060\textwidth}>{\raggedright\arraybackslash}p{0.1\textwidth}>{\raggedright\arraybackslash}p{0.15\textwidth}>{\raggedright\arraybackslash}p{0.22\textwidth}>{\raggedright\arraybackslash}X}
\toprule
\textbf{Use Case} & \textbf{RL Type} & \textbf{Training Environment} & \textbf{Network Scenario Variables} & \textbf{Reward (performance)}: $J(\pi;e)$ & \textbf{Reference Policy} \\
\midrule
CCA (Sage~\cite{yen2023sage}) & Offline & Emulation & Bandwidth, loss, delay &
\(\begin{gathered}\sum_t\left(\alpha\cdot\,\mathrm{rate}_t+\beta\cdot\,\mathrm{rtt}_t-\gamma\cdot\,\mathrm{loss}_t\right)\end{gathered}\) &
Best performance among Cubic, BBR, and NewReno. \\
ABR (Pensieve~\cite{mao_neural_2017_pensieve}) & Online & Simulation & Bandwidth &
\(\begin{gathered}\sum_t\left(\alpha\cdot b_t-\beta\cdot r_t-\gamma\cdot\lvert b_t-b_{t-1}\rvert\right)\end{gathered}\) &
Exact $K$-step rolling bitrate oracle. \\
LB (Park~\cite{mao2019park}) & Online & Simulation & Job arrival time, job size &
$-\alpha\cdot\int_0^T N_{\mathrm{active}}(t)\,dt$ &
Best performance among LCT, JSQ, and CFS. \\
\bottomrule
\end{tabularx}
\end{adjustbox}
\caption{Implementation summary for the instantiations of Eqn.~\eqref{eq:outer_hat} for the three use cases described in Appendix~\ref{app:use-cases}.
For Sage, $\alpha=0.60$, $\beta=0.25$, and $\gamma=0.15$; $\mathrm{rate}_t$ is delivery rate divided by path capacity and clipped to $[0,1]$, $\mathrm{rtt}_t$ is base RTT divided by current RTT and clipped to $[0,1]$, and $\mathrm{loss}_t$ is loss normalized by path capacity and clipped to $[0,1]$.
For Pensieve, $b_t$ is bitrate in Kbps, $r_t$ is rebuffering time, $\alpha=\gamma=10^{-3}$, and $\beta=4.3$.
For Park, reward is negative active-job-time cost, so maximizing reward minimizes accumulated active-job time.
$N_{\mathrm{active}}(t)$ is the number of active jobs, $T$ is the episode horizon, and $\alpha=1/\texttt{reward\_time\_scale}$.
LCT denotes Least Completion Time, JSQ denotes Join Shortest Queue, and CFS denotes Choose Fastest Server.}
\label{tab:maxregret_impl}
\end{table*}

\mypar{High-level design}
Table~\ref{tab:maxregret_impl} summarizes how \sys instantiates Eqn.~\eqref{eq:outer_hat} across three use cases detailed in Appendix~\ref{app:use-cases}: congestion control, adaptive bitrate streaming, and load balancing.
Each implementation exposes a small set of scenario variables to the outer RL solver, keeps the pretrained controller $\pi$ fixed, and constructs a practical reference policy $\hat{\pi}_e^\star$ under the same generated scenario.
At time $t$, the solver observes the controller-induced system state, chooses the next resource element $u_t$, and appends it to the scenario $e=(u_0,\dots,u_{T-1})$.
\sys then evaluates $\pi$ and the reference under this shared scenario and uses the resulting regret from Eqn.~\eqref{eq:outer_hat} as the learning signal.

To ensure a fair and diagnostic comparison across policies, \sys evaluates every policy under the exact same exogenous scenario while keeping each policy's closed-loop state isolated.
Using the application-specific reward in Table~\ref{tab:maxregret_impl}, the outer solver is rewarded only for finding conditions where $\pi$ performs poorly while $\hat{\pi}_e^\star$ still achieves high reward.
Each implementation defines $\hat{\pi}_e^\star$ as the best policy in its practical reference class:
\begin{equation}
\label{eq:impl_reference}
J(\hat{\pi}_e^\star;e)
:=
\operatorname*{maximize}_{\pi' \in \Pi}\ J(\pi';e).
\end{equation}
Eqn.~\eqref{eq:impl_reference} defines the reference used by the implementation over its chosen $\Pi$, approximating the scenario-specific optimum in Eqn.~\eqref{eq:bilevel_true}.
For each application, $\Pi$ is a portfolio of complementary heuristics whose upper envelope serves as a strong practical reference~\cite{pappone2025mutant,yen2023sage,agarwal2025mowgli}, avoiding dependence on any single algorithm's blind spots.
As summarized in Table~\ref{tab:maxregret_impl}, the portfolios pair delay- and loss-based controllers for Sage (BBR, Cubic, NewReno), buffer- and rate-based policies for Pensieve (BBA, RobustMPC, rolling oracle), and load- and duration-aware heuristics for Park (LLF, SJF, LCT).
A stronger $\hat{\pi}_e^\star$ sharpens the reward signal for the outer solver and reduces $\delta(e)$, tightening the certificate in Theorem~\ref{thm:exact}.

\mypar{Congestion control}
Instantiating \sys for Sage~\cite{yen2023sage} (Fig.~\ref{fig:sage_implementation}, with Pensieve and Park detailed in Appendix~\ref{app:maxregret-implementations}) requires process-level emulation because \sys evaluates a real Linux congestion-control stack rather than a standalone simulator.
\sys extends Mahimahi~\cite{netravali2015mahimahi} from a passive trace player into an online-controlled, multi-policy emulator.
At each outer step, $u_t$ specifies the bottleneck link conditions for the next control interval.
The common training mode maps a normalized solver action logarithmically to a shared bottleneck bandwidth, while richer configurations also control loss and propagation delay.
This logarithmic bandwidth mapping gives the solver fine resolution in both low-capacity and high-capacity regimes because transport behavior scales with bandwidth ratios rather than absolute Mbps increments.

Sage uses a best-of-portfolio reference over conventional congestion controllers.
\sys runs Sage, BBR, Cubic, and NewReno under the same link evolution, then applies Eqn.~\eqref{eq:impl_reference} using the normalized transport reward in Table~\ref{tab:maxregret_impl}.
The reward combines delivery rate relative to available capacity, RTT inflation relative to the base RTT, and packet loss normalized by capacity, while leaving multi-flow fairness to future work.
Normalizing these terms makes scores comparable across the full bandwidth range that the solver generates.
The objective in Eqn.~\eqref{eq:outer_hat} then favors selective performance failures, such as bandwidth transitions, latency regimes, or loss bursts where Sage underreacts, overreacts, or recovers sluggishly while at least one conventional controller remains effective.

The central challenge in Sage is synchronized multi-policy emulation.
\sys launches one isolated Mahimahi child instance per policy, with separate ports, actor identifiers, runtime directories, TCP state, queue dynamics, and shared-memory observation channels.
Isolation prevents one controller's packets, queues, or control history from contaminating another controller's trajectory.
At the same time, \sys coordinates these child instances as a single logical experiment by applying the same $u_t$ to every child at the same logical step and effective timestamp.
This lockstep execution ensures that every policy experiences the exact same link change before the performance gap is scored.

The Mahimahi extension uses an online control plane that updates per-direction bandwidth, delay, and loss dynamically while returning detailed queue and departure telemetry.
\sys monitors this telemetry to detect stale settings or desynchronized updates and aborts inconsistent steps rather than scoring misleading gaps.
Inside each solver interval, \sys also supports deterministic loss over short sub-intervals, allowing the solver to express bursty or temporally correlated loss patterns while keeping $u_t$ compact.


\section{Protecting RL Controllers against Performance Failures}
\label{sec:shield}

\begin{figure}[t]
    \centering
    \begin{adjustbox}{width=1\linewidth,center=0pt}
    \includegraphics[width=\linewidth]{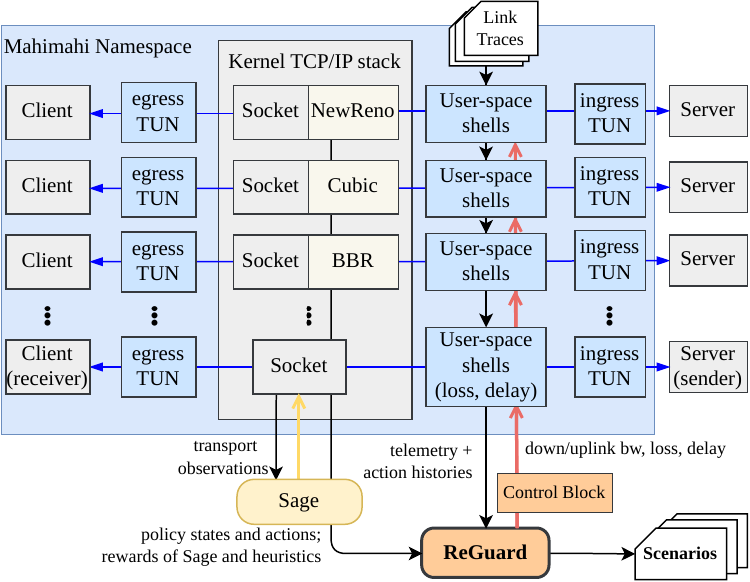}
    \end{adjustbox}
\caption{The key design point of Sage's instantiation is end-to-end counterfactual consistency: \sys runs Sage and the reference controllers in synchronized but isolated emulation, then reuses the resulting controller-reference trajectories to identify risky states and derive small corrections.}
\label{fig:sage_implementation}
\vspace{-4pt}
\end{figure}

\sys uses the scenarios found by the search in \S\ref{subsec:maxregret_impl} to learn lightweight inference-time protection for the pretrained controller.
Rather than retraining the controller or running reference heuristics online, \sys distills controller-reference comparisons from challenging scenarios into interpretable rules that decide when to trust the controller and when to apply a bounded correction.
This design turns failure discovery directly into runtime protection: the search exposes where the controller leaves reward on the table, and \sys converts those failures into targeted interventions that preserve nominal behavior elsewhere.

\subsection{Identifying Risky Control States}

Failure discovery alone does not reveal which individual decisions caused the performance gap.
A high-regret scenario shows that the controller underperforms across an episode, but it does not specify which states are risky, whether the controller acted too aggressively or too conservatively, or how the action should change.
\sys bridges this gap with a \emph{Counterfactual Analyzer} (Fig.~\ref{fig:overview}), which treats the reference-policy portfolio as a set of counterfactual teachers.
For each decision point $t$, the analyzer records the controller state $s_t$, the controller action $a_t=\pi(s_t)$, the reference action induced by $\hat{\pi}_e^\star$, and the rewards achieved under the same scenario.

The Counterfactual Analyzer converts trajectory-level failures into state-level supervision.
It identifies states where the controller is consistently suboptimal, compares the controller action with the reference action under the same observed state, labels the corrective direction, and summarizes recurring discrepancies as operator-meaningful patterns.
Specifically, a visited state $s_t$ in scenario $e$ is marked risky only when the episode regret exceeds a threshold $\tau_{\mathrm{regret}}$ and the reference policy outperforms the controller from that state:
\begin{equation}
\label{eq:risky_state}
\mathrm{Risky}(s_t) = 1 \iff J(\hat{\pi}_e^\star;e) - J(\pi;e) \ge \tau_{\mathrm{regret}}.
\end{equation}
This criterion keeps \sys's protection focused on avoidable controller mistakes rather than intrinsically harsh conditions where every policy struggles.

Comparing the controller against successful reference trajectories under identical conditions surfaces concrete decision-time failure modes: Pensieve selecting bitrates that recent throughput and buffer levels cannot sustain, Sage backing off too sharply after bandwidth drops and reopening too slowly, or Park routing small jobs to slow servers while fast servers sit idle.
These comparisons pinpoint systematic decision-time symptoms and provide the corrective signal needed for runtime intervention.

The same counterfactual comparison provides a discrete corrective label from a compact, application-independent set:
\begin{equation}
\label{eq:adjustment_set}
\mathcal{C}_{\mathrm{adjustment}}
:=
\{\textsc{abstain},\textsc{back\_off},\textsc{push\_harder}\}.
\end{equation}
\textsc{abstain} leaves the controller action unchanged, \textsc{back\_off} indicates that the controller acted too aggressively for the observed state, and \textsc{push\_harder} indicates that the controller acted too conservatively.
When the regret is small, the reference is not reliably better, or the corrective direction is ambiguous, the Counterfactual Analyzer outputs \textsc{abstain}.
The heuristic portfolio thus serves as the source of counterfactual action labels that make runtime protection actionable.

\subsection{Learning Protection Rules from Risky States}

\sys learns threshold-based logic rules that map observable controller state regimes directly to the adjustment labels produced by the Counterfactual Analyzer.
To synthesize these rules, \sys adapts the rule-learning stage of NetNomos~\cite{he2026netnomos} to associate controller state predicates with corrective action labels.

\mypar{Predicate construction}
\sys constructs its predicate vocabulary from normal controller behavior by converting tabular state observations into symbolic threshold predicates~\cite{kruit2020tab2know,he2026autogram}.
It runs the controller on normal traces, computes percentile thresholds for each observable state variable, and forms predicates such as $x_i(s) \ge q_{i,95}^{\mathrm{normal}}$, $x_i(s) \le q_{i,25}^{\mathrm{normal}}$, or $x_i(s) \ge q_{i,90}^{\mathrm{normal}}$, where $q_{i,p}^{\mathrm{normal}}$ is the $p$th percentile of feature $i$ on normal traces.
Expressing each condition relative to the controller's nominal operating distribution keeps the predicates interpretable across domains~\cite{he2024peros}, capturing high loss or RTT inflation for Sage, low buffer occupancy or high download delay for Pensieve, and load imbalance or large incoming jobs for Park.

\mypar{Rule synthesis}
NetNomos receives the percentile predicates and the Counterfactual Analyzer's adjustment labels and synthesizes implications of the form
\begin{equation}
\label{eq:rule_implication}
\phi_1(s_t) \wedge \phi_2(s_t) \wedge \cdots \wedge \phi_k(s_t) \implies \mathrm{adjustment}(s_t) = c,
\end{equation}
where each $\phi_j(s_t)$ is a percentile threshold predicate on a single state feature and $c \in \mathcal{C}_{\mathrm{adjustment}}$.
Each learned rule prescribes a single adjustment whenever its predicate conjunction holds.
This simple rule structure keeps the deployed mechanism compact and inspectable, giving operators a clear explanation of when the controller fails and how its action is adjusted.

\mypar{Iterative refinement}
Because an initial rule set may miss secondary failure states that were masked by the most severe initial flaws, \sys refines its rules through an iterative dataset-aggregation loop.
Each round runs scenario discovery against the \emph{currently protected} controller, surfacing counterexample scenarios that bypass or stress the current rules.
\sys replays these new scenarios, uses the Counterfactual Analyzer to label the newly visited states, and appends the resulting examples to a cumulative dataset.

In each round, \sys re-learns the rule set \emph{from scratch} over the aggregated dataset rather than appending ad hoc patches to the previous rules.
Synthesizing a single coherent rule set over all evidence collected so far avoids logical conflicts between old and new clauses, prevents rule explosion, and allows the learner to tighten or relax earlier predicate boundaries as new counterexamples arrive.

This iterative loop directs search capacity toward remaining blind spots while retaining earlier examples of where to intervene and where to abstain, preventing overfitting to the most recent round.
As a result, refinement broadens coverage of risky regimes while keeping the deployed rule set compact and consistent.

\subsection{Efficient Inference-time Rule Enforcement}

\sys enforces its learned rules at inference time using fast, memoryless predicate evaluation over the current observable state $s_t$~\cite{he2025lejit}.
At each decision step $t$, \sys evaluates predicates on $s_t$ and adjusts only the controller's proposed action, without tracking full trajectories, running reference policies, or solving an optimization problem online.

This memoryless design matches the standard RL observation interface.
RL controllers already summarize recent history, including throughput samples, delay measurements, prior actions, and queue estimates, inside the current state vector $s_t$.
Under the Markov property~\cite{sutton2018reinforcement,mnih2015human,jay2019deep,dong2018pcc,mccallum1995reinforcement}, $s_t$ contains the state features needed for the next decision, so evaluating rules over $s_t$ suffices to detect risky regimes and adjust the proposed action.

Formally, let $\mathcal{R}$ be the learned rule set and let $\pi(s_t)$ be the action proposed by the pretrained controller.
\sys computes
$$
c_t := g_{\mathcal{R}}(s_t)
\quad\text{and}\quad
a_t^{\mathrm{prot}} := h(\pi(s_t), c_t, s_t),
$$
where $g_{\mathcal{R}}$ maps the current state to an adjustment label and $h$ maps that label to an application-specific bounded correction.
If $c_t=\textsc{abstain}$, then $a_t^{\mathrm{prot}}:=\pi(s_t)$.
If $c_t$ is \textsc{back\_off} or \textsc{push\_harder}, then $h$ nudges or masks the action in the direction supported by the Counterfactual Analyzer's labels.
The RL controller remains the primary decision maker because \sys intervenes only when a learned rule fires.

To minimize runtime latency, \sys organizes the conjunctions in $\mathcal{R}$ into a prefix trie~\cite{beckett2026concord}.
Rules that share predicate prefixes share a common path in the trie, allowing \sys to prune entire subtrees as soon as a predicate evaluates to false.
Lookup cost therefore scales with the depth of matched predicates rather than the total number of rules.

This architecture creates a sharp computational asymmetry between offline discovery and online protection.
While offline scenario search and counterfactual analysis use cluster compute to train outer RL solvers and run multi-policy emulations, online enforcement executes only scalar comparisons over already-computed state features and occupies a few kilobytes of memory.
As a result, all protected controllers easily meet their real-time decision deadlines (Fig.~\ref{fig:controller_latency_plots}).

\subsection{Implementation}

\mypar{Congestion control}
For Sage~\cite{yen2023sage} (with Pensieve and Park protection detailed in Appendix~\ref{app:shield-implementations}), the protected controller observes the 69-dimensional state vector summarized in Table~\ref{tab:sage_state_features}, covering RTT, RTT variation, delivery rate, loss, minimum-RTT ratios, and prior actions.
Because Sage outputs a continuous rate-scaling action $a_t \in [a_{\min}, a_{\max}]$, \sys implements its adjustments as bounded nudges of step size $\Delta > 0$ along the primary control dimension:
\begin{equation}
\label{eq:sage_nudge}
a_t^{\mathrm{prot}} :=
\begin{cases}
a_t, & \text{if } c_t = \textsc{abstain}, \\
\mathrm{clip}(a_t - \Delta,\; a_{\min},\; a_{\max}), & \text{if } c_t = \textsc{back\_off}, \\
\mathrm{clip}(a_t + \Delta,\; a_{\min},\; a_{\max}), & \text{if } c_t = \textsc{push\_harder}.
\end{cases}
\end{equation}
This rule-guided nudge preserves Sage as the base controller and adjusts its aggressiveness only when the rules detect a verified risky state.

\section{Evaluation} \label{sec:eval}

\begin{figure*}[t]
    \centering
    \begin{subfigure}[t]{0.33\textwidth}
        \centering
        \includegraphics[width=\linewidth]{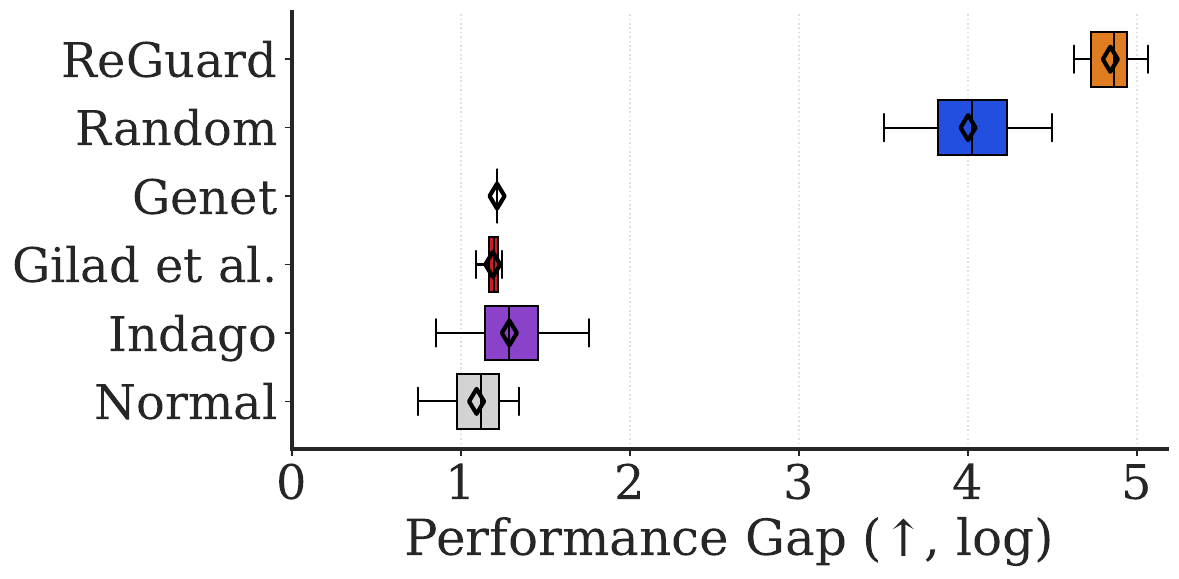}
        \caption{Pensieve.}
        \label{fig:pensieve_gap}
    \end{subfigure}
    \hfill
    \begin{subfigure}[t]{0.33\textwidth}
        \centering
        \includegraphics[width=\linewidth]{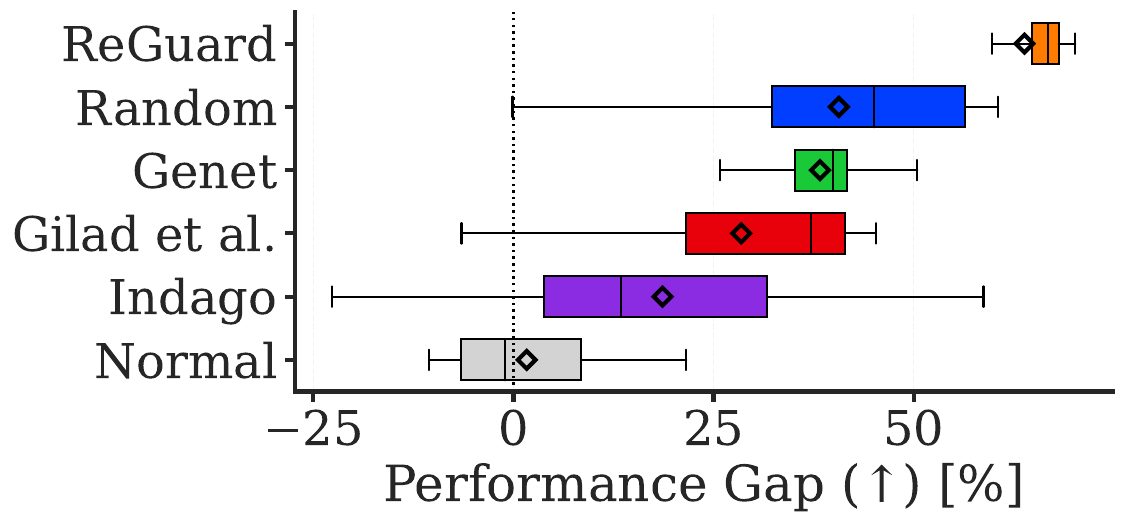}
        \caption{Sage.}
        \label{fig:sage_gap}
    \end{subfigure}
    \hfill
    \begin{subfigure}[t]{0.33\textwidth}
        \centering
        \includegraphics[width=\linewidth]{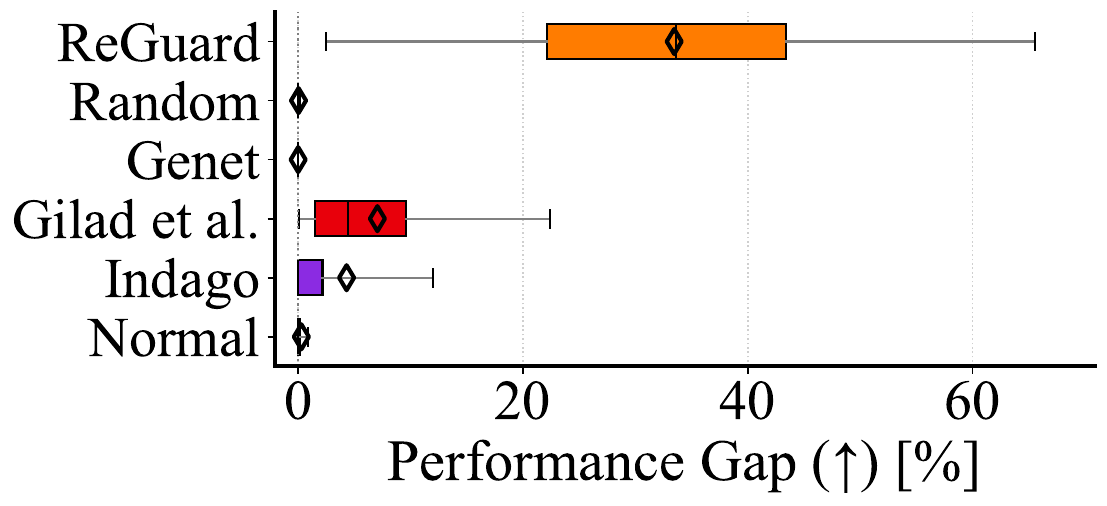}
        \caption{Park.}
        \label{fig:park_gap}
    \end{subfigure}
    \caption{\sys finds the most challenging scenarios in all three use cases.
    The gap is consistently larger than the strongest baseline and far above the Normal regime, showing that \sys exposes large avoidable underperformance rather than marginally harder tests.}
    \label{fig:gap_plots}
\end{figure*}

\begin{figure*}[t]
    \centering
    \begin{subfigure}[t]{0.33\textwidth}
        \centering
        \includegraphics[width=\linewidth]{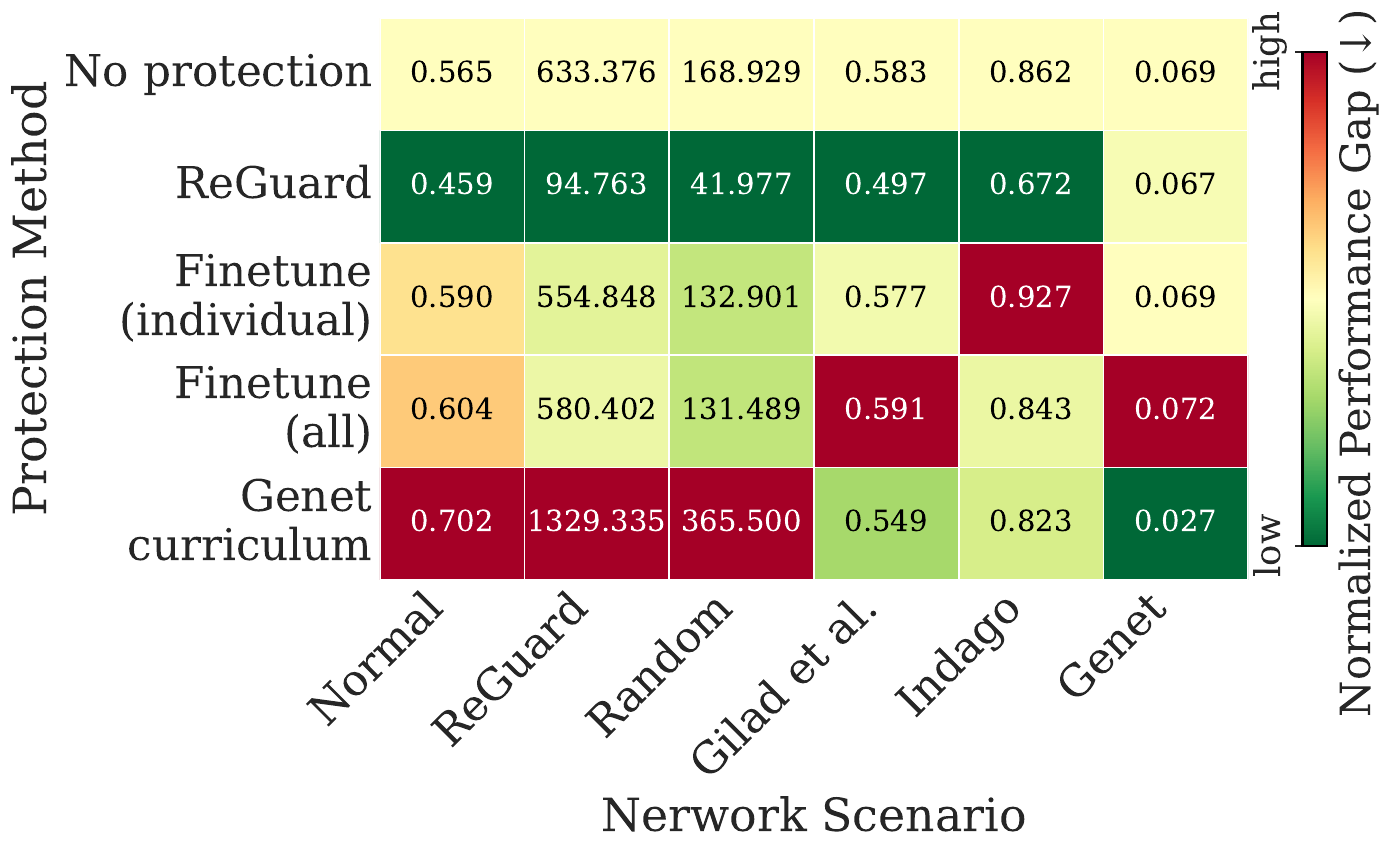}
        \caption{Pensieve.}
        \label{fig:pensieve_heatmap}
    \end{subfigure}
    \hfill
    \begin{subfigure}[t]{0.33\textwidth}
        \centering
        \includegraphics[width=\linewidth]{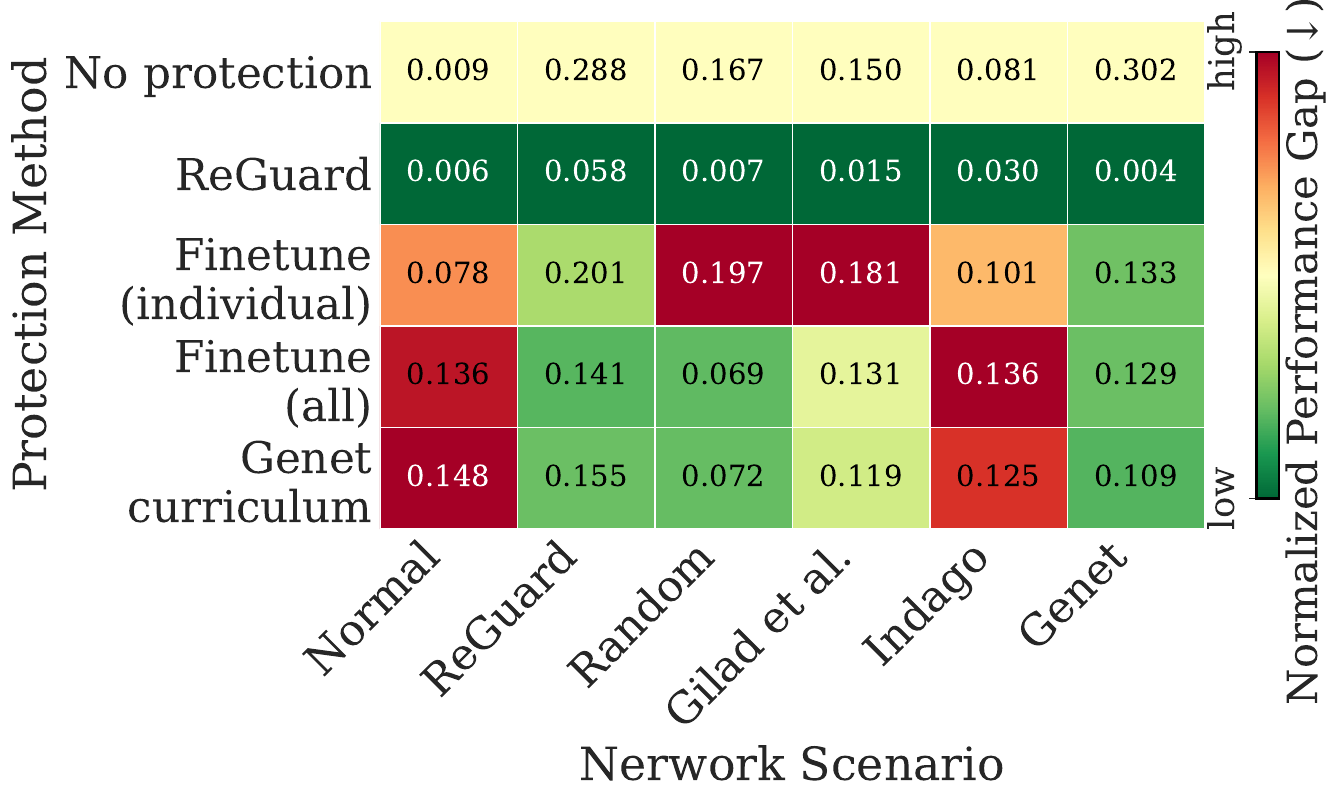}
        \caption{Sage.}
        \label{fig:sage_heatmap}
    \end{subfigure}
    \hfill
    \begin{subfigure}[t]{0.33\textwidth}
        \centering
        \includegraphics[width=\linewidth]{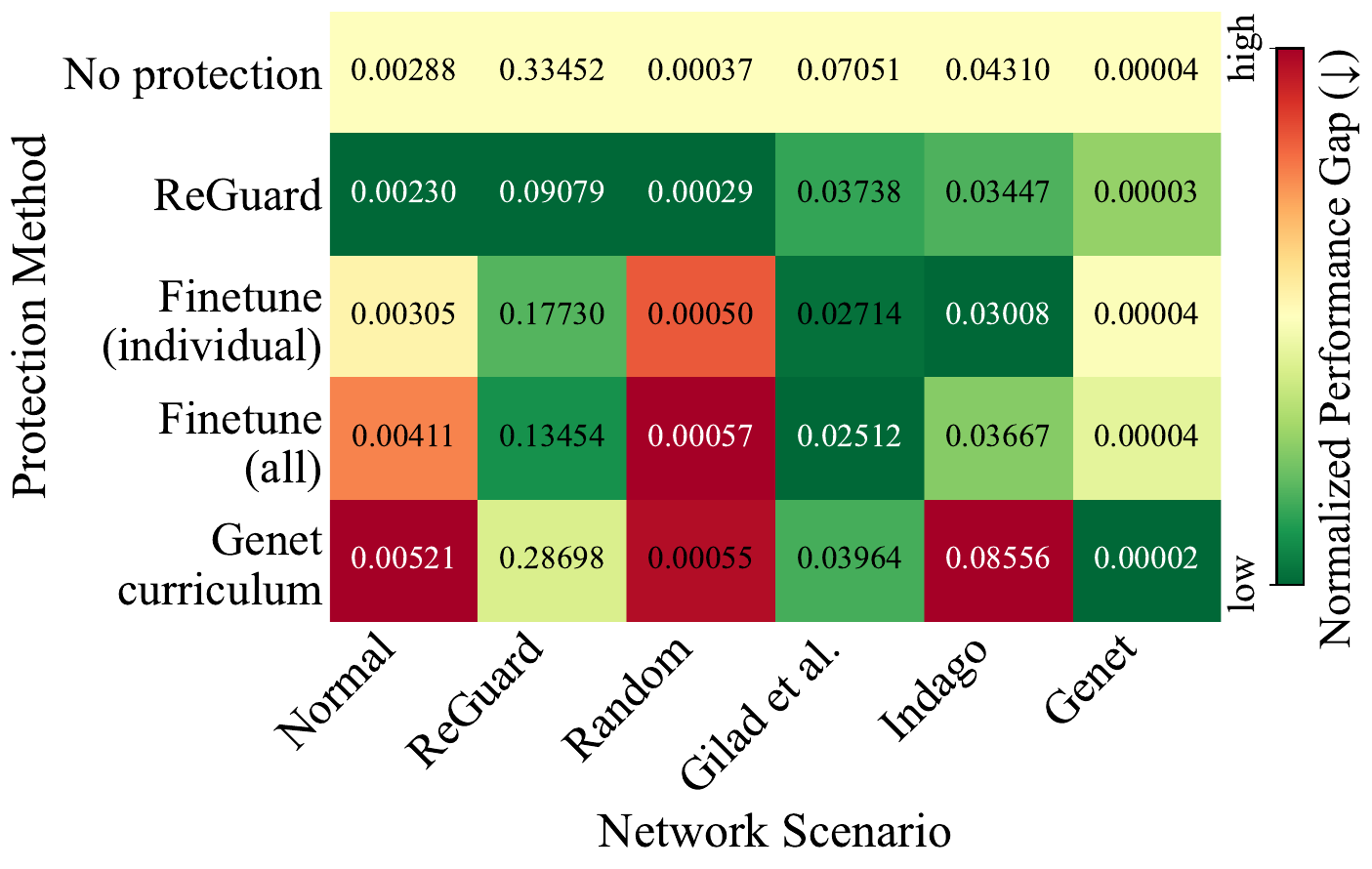}
        \caption{Park.}
        \label{fig:park_heatmap}
    \end{subfigure}
    \caption{\sys provides the strongest overall protection when derived from \sys scenarios, while preserving nominal performance.
    The results of \sys shown here come from its \emph{very first iteration}, and later iterations provide even more protection (see Fig.~\ref{fig:iterations}).
    \sys is most effective on the most challenging scenarios, but it also transfers to other, less challenging scenario families, indicating that it captures recurring risky states rather than memorizing patterns from one method.}
    \label{fig:heatmaps}
\end{figure*}

\begin{figure*}[t]
    \centering
    \begin{subfigure}[t]{0.33\textwidth}
        \centering
        \includegraphics[width=\linewidth]{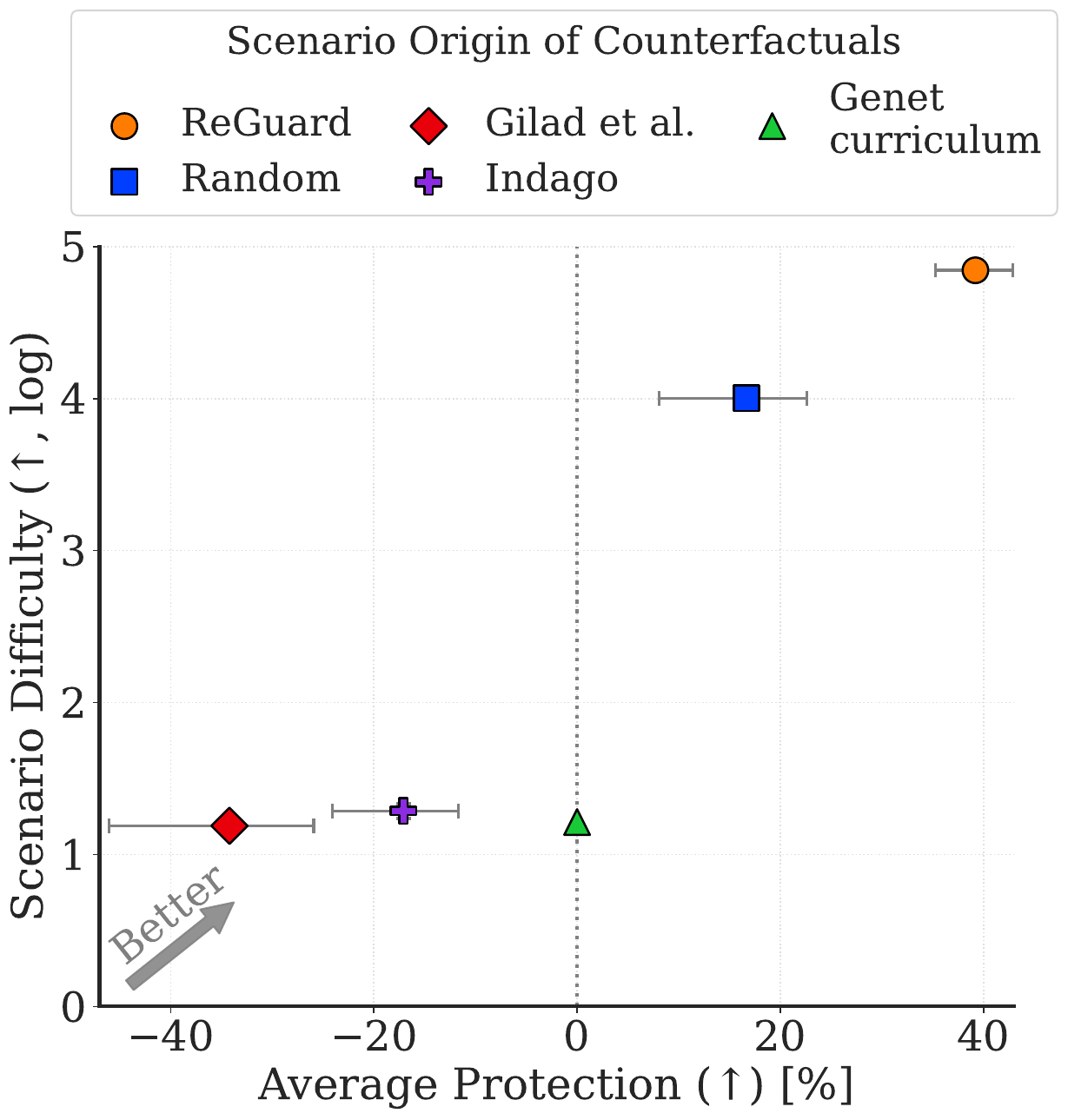}
        \caption{Pensieve.}
        \label{fig:pensieve_ablation}
    \end{subfigure}
    \hfill
    \begin{subfigure}[t]{0.33\textwidth}
        \centering
        \includegraphics[width=\linewidth]{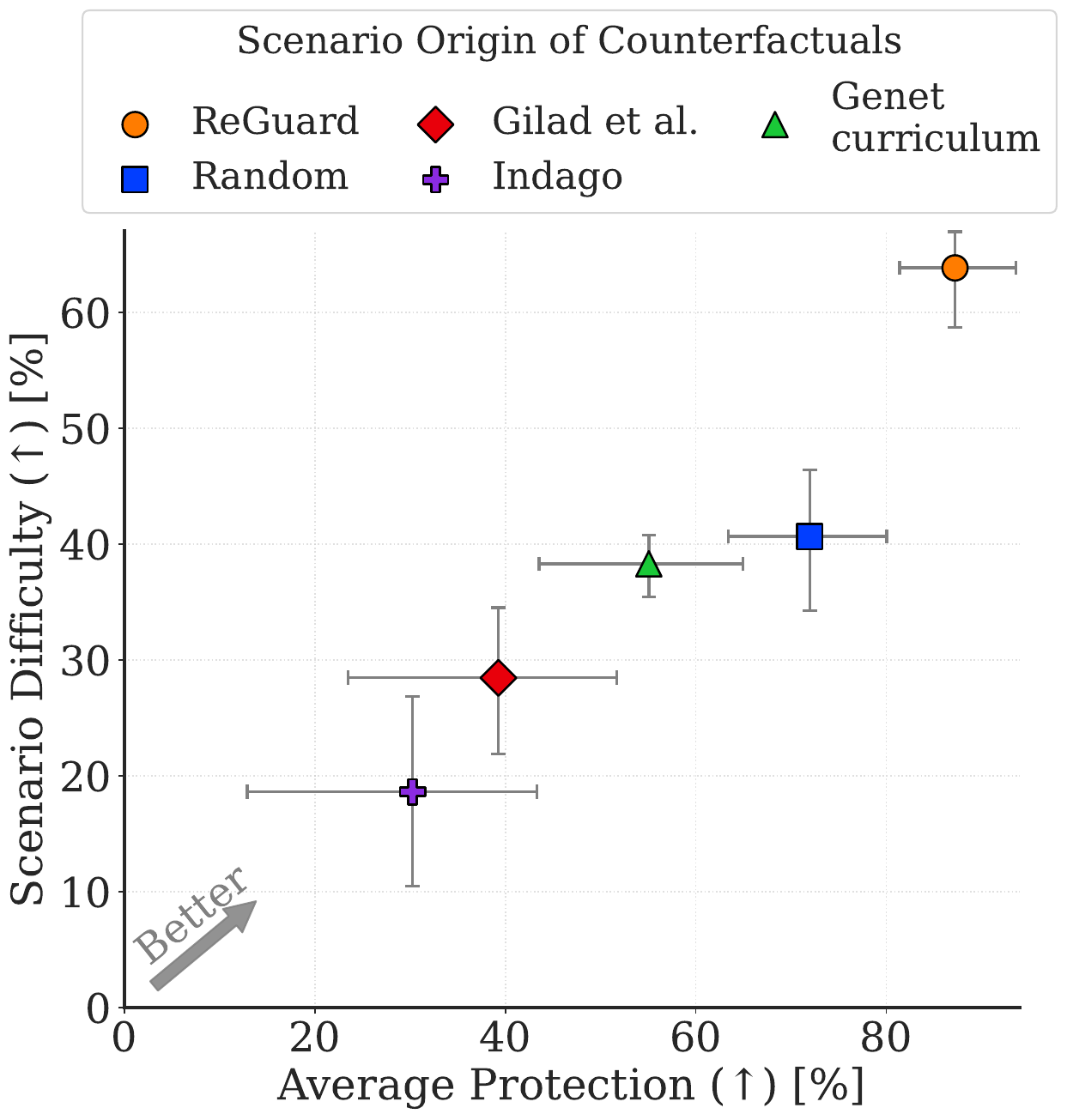}
        \caption{Sage.}
        \label{fig:sage_ablation}
    \end{subfigure}
    \hfill
    \begin{subfigure}[t]{0.33\textwidth}
        \centering
        \includegraphics[width=\linewidth]{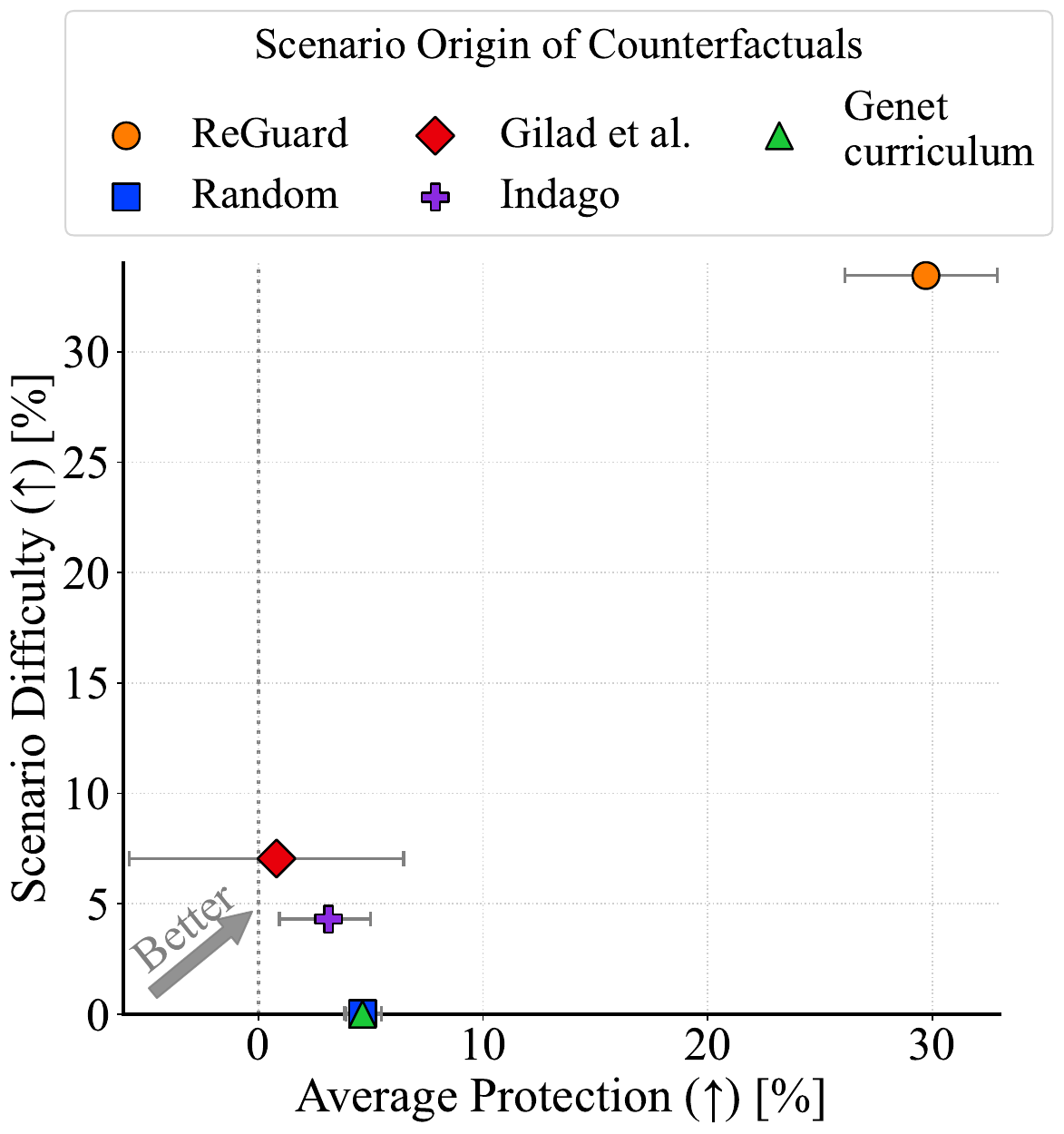}
        \caption{Park.}
        \label{fig:park_ablation}
    \end{subfigure}
    \caption{Learning protection rules from higher-regret scenarios yields stronger runtime protection across all three controllers.
    Search quality therefore matters directly: discovering severer failures produces more effective protection rules.}
    \label{fig:ablation}
\end{figure*}

\begin{figure*}[t]
    \centering
    \begin{subfigure}[t]{0.33\linewidth}
        \centering
        \includegraphics[width=\linewidth]{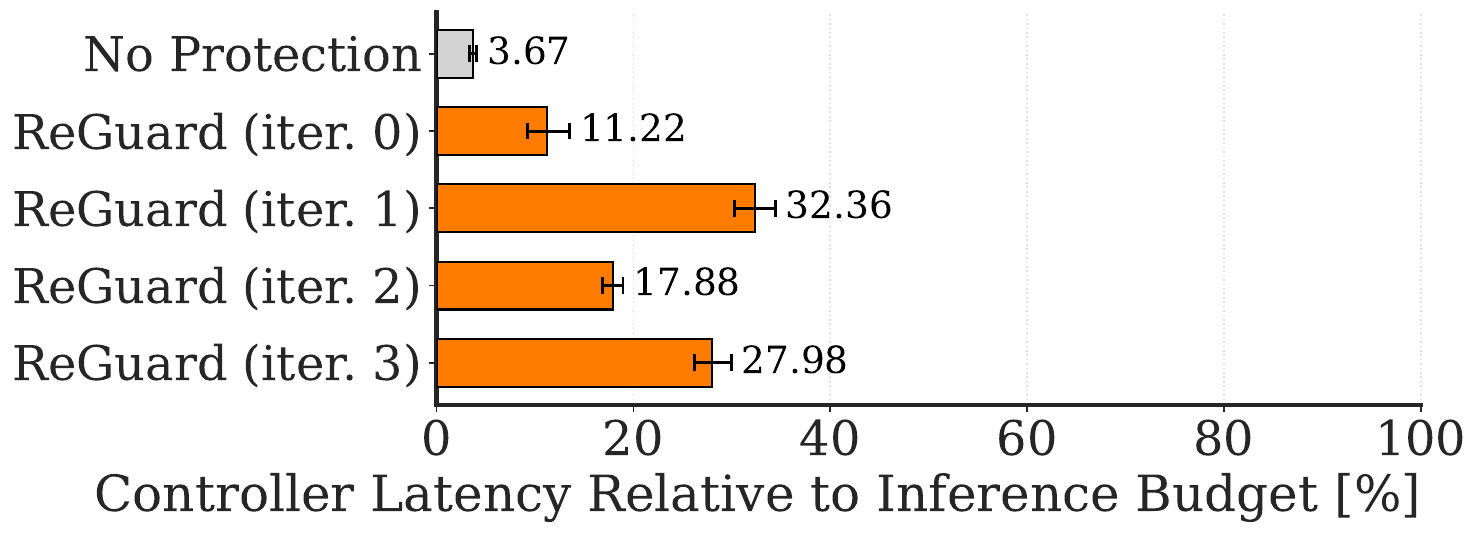}
        \caption{Pensieve.}
        \label{fig:pensieve_latency}
    \end{subfigure}
    \hfill
    \begin{subfigure}[t]{0.33\linewidth}
        \centering
        \includegraphics[width=\linewidth]{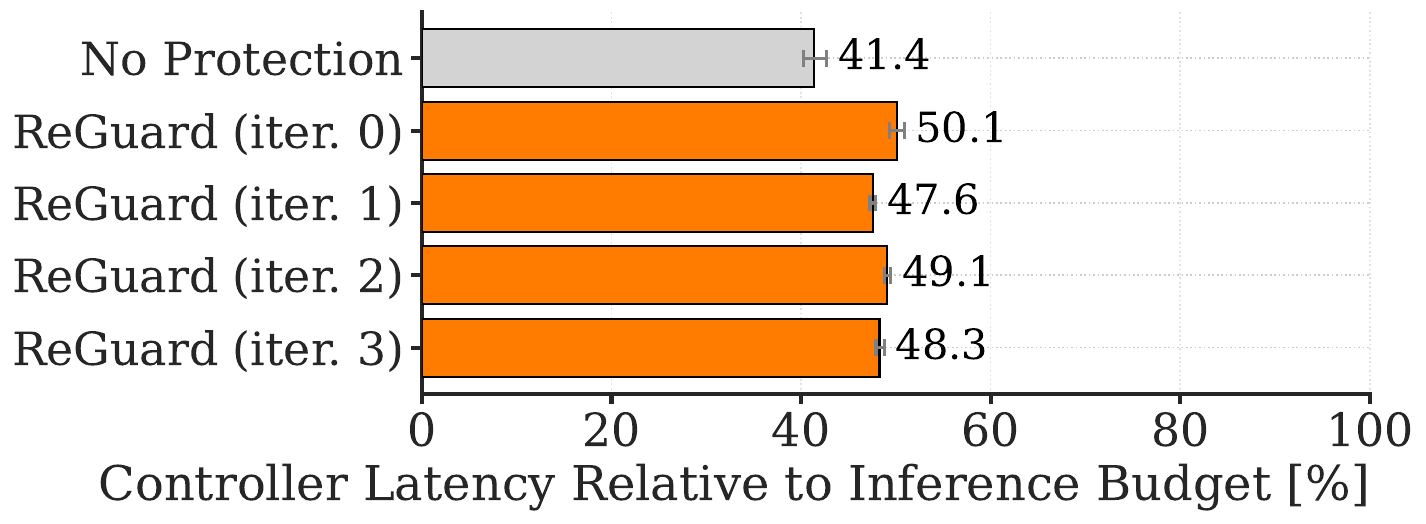}
        \caption{Sage.}
        \label{fig:sage_latency}
    \end{subfigure}
    \hfill
    \begin{subfigure}[t]{0.33\linewidth}
        \centering
        \includegraphics[width=\linewidth]{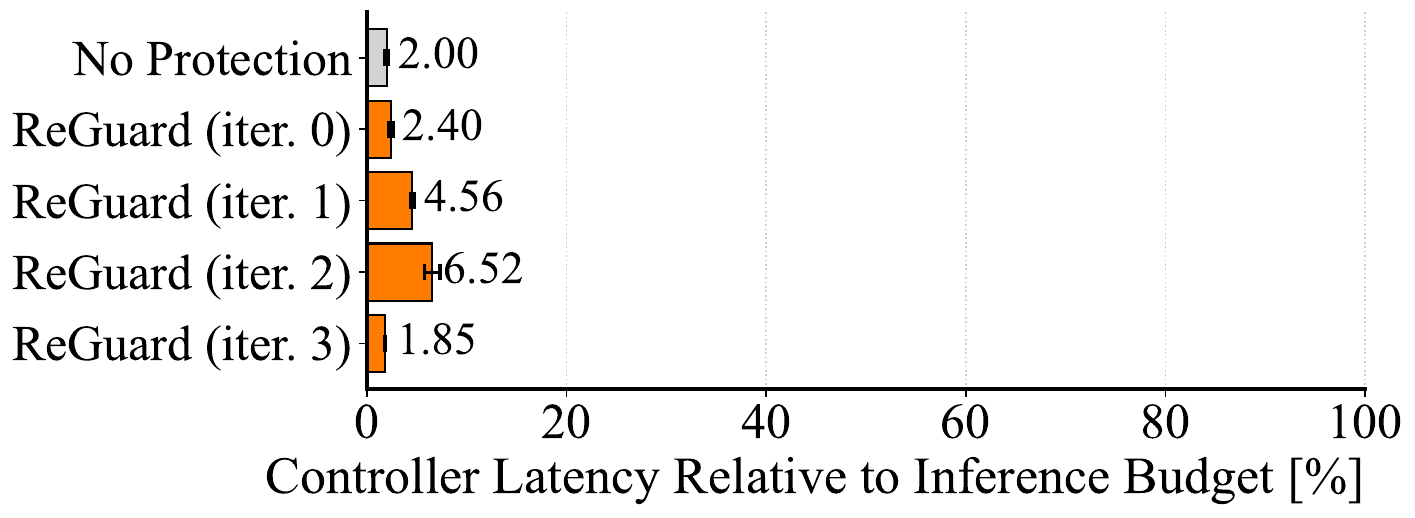}
        \caption{Park.}
        \label{fig:park_latency}
    \end{subfigure}
    \caption{\sys remains well within the online decision budget in all three systems across all refinement iterations.
    Later iterations strengthen protection, but they do not introduce systematic latency growth.}
    \label{fig:controller_latency_plots}
\end{figure*}

\begin{figure*}[t]
    \centering
    \begin{subfigure}[t]{0.33\linewidth}
        \centering
        \includegraphics[width=\linewidth]{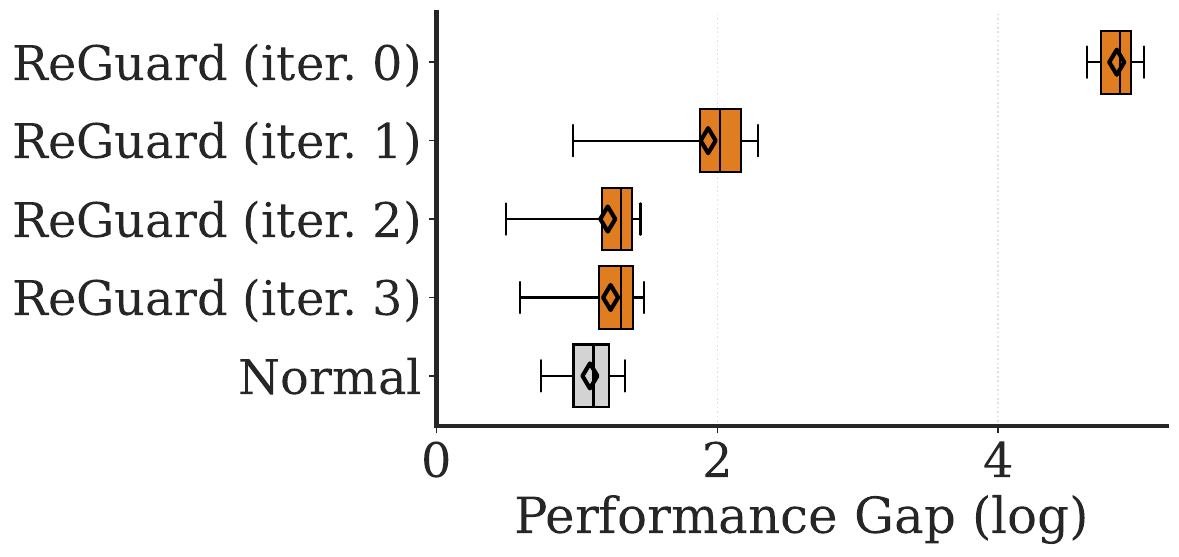}
        \caption{Pensieve.}
        \label{fig:pensieve_iters}
    \end{subfigure}
    \hfill
    \begin{subfigure}[t]{0.33\linewidth}
        \centering
        \includegraphics[width=\linewidth]{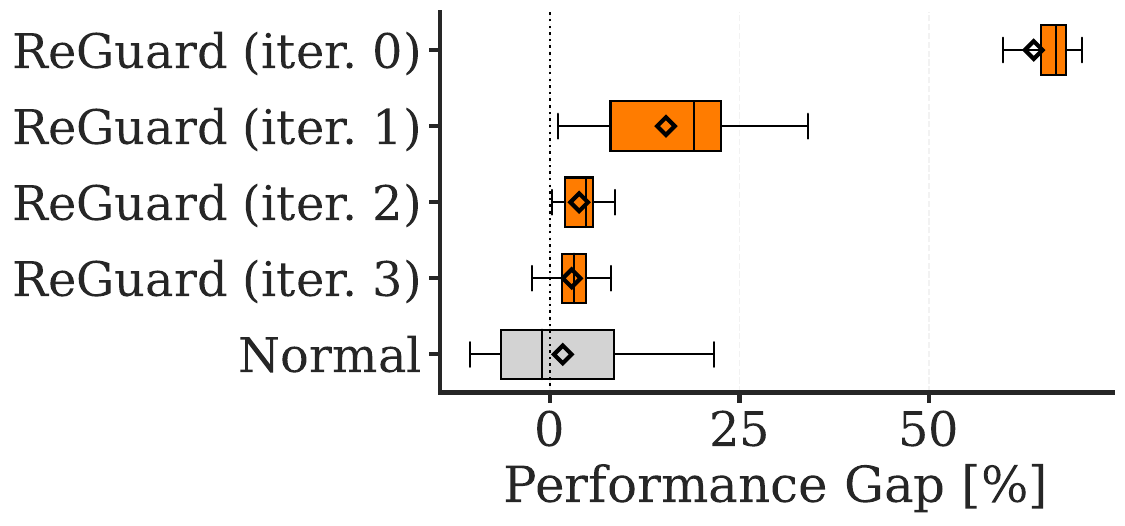}
        \caption{Sage.}
        \label{fig:sage_iters}
    \end{subfigure}
    \hfill
    \begin{subfigure}[t]{0.33\linewidth}
        \centering
        \includegraphics[width=\linewidth]{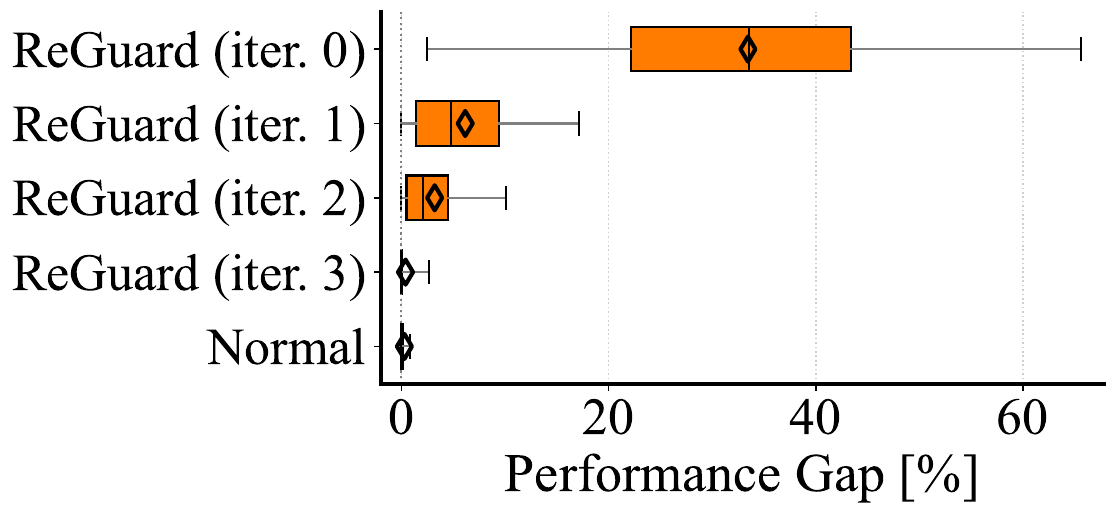}
        \caption{Park.}
        \label{fig:park_iters}
    \end{subfigure}
    \caption{A few search-and-protect iterations are enough to drive the discovered gap close to the Normal regime.
    Most of the improvement arrives in the first two rounds, after which the remaining challenging cases are much less severe.}
    \label{fig:iterations}
\end{figure*}

We prototype \sys~\cite{reguard_codebase} and aim to answer the following evaluation questions through experiments:
\begin{enumerate}[nosep, label=E\arabic*., ref=E\arabic*,]
    \item\label{e1} Can \sys find network scenarios where a given RL controller incurs a larger performance gap than competing baselines expose?
    \item\label{e2} Can \sys effectively protect a given RL controller without retraining while preserving its nominal performance?
    \item\label{e3} How does the severity of the scenarios used to learn protection rules affect the effectiveness of \sys's protection?
    \item\label{e4} How much overhead does \sys impose on the RL controller during inference?
    \item\label{e5} Can \sys iteratively reduce the performance gap of the most challenging network scenarios to the level of normal scenarios?
\end{enumerate}

\subsection{Setup}

\mypar{Network Traces}
We leverage diverse network and workload traces (referred to as ``Normal'' in our experiments) to train and evaluate controllers across intended, diverse operating regimes~\cite{ustiugov2023invitro,jin2027tracebleed}:
\begin{itemize}[leftmargin=*, nosep]
    \item Sage~\cite{yen2023sage} is evaluated on intra-continental live Internet paths across 16 U.S. cities, intercontinental paths across 13 international servers, and 23 cellular traces gathered in NYC, following prior train-test splits~\cite{abbasloo2020orca,yen2023sage}.
    \item Pensieve~\cite{mao_neural_2017_pensieve} is evaluated on 375 FCC broadband traces~\cite{fcctraces} and 425 Norway cellular traces~\cite{norwaytraces} under the train-test split from Genet~\cite{xia2022genet}.
    \item Park~\cite{mao2019park} is evaluated on 200 traces across 5 workload distributions (``Default'', ``Original'', ``RL1'', ``RL2'', and ``RL3'')~\cite{xia2022genet}, split evenly between training and testing.
\end{itemize}

\mypar{Baselines for discovering challenging network scenarios}
\begin{itemize}[leftmargin=*, nosep]
    \item \textbf{Indago}~\cite{biagiola2024indago} trains a surrogate failure predictor and uses saliency gradients to search for failure-inducing scenarios.
    \item \textbf{Genet}~\cite{xia2022genet} generates curricula maximizing raw gap-to-baseline (using BBR~\cite{cardwell2017bbr} for CCA, RobustMPC~\cite{yin2015robustmpc} for ABR, and least-load-first for LB).
    \item \textbf{Gilad et al.}~\cite{gilad_robustifying_2019} directly searches for network conditions that degrade controller performance.
    \item \textbf{Random} uniformly fuzzes network scenario variables across their entire domains (Table~\ref{tab:maxregret_impl}).
\end{itemize}

\mypar{Protection baselines}
To protect the RL controller against challenging network scenarios, we compare \sys\footnote{The logic rules are always derived from training traces, never from held-out test traces.} against the following retraining-based approaches:
\begin{itemize}[leftmargin=*, itemsep=0pt, topsep=2pt]
    \item \textbf{Fine-tuning}: The most direct way to improve an RL policy on a given workload family is to continue training it on that workload.
    We consider two fine-tuning settings:
    \textbf{(individual)} We fine-tune the RL controller on traces from discovered scenarios in one specific family, and evaluate the resulting policy on held-out test traces of the \emph{same} kind (\eg fine-tuning on Indago's traces and then testing on Indago's held-out traces).
    \textbf{(all)} We fine-tune the RL controller on traces from \emph{all} scenario families (namely, \sys, Indago, Genet, Gilad et al., and Random) and evaluate the resulting policy separately on held-out traces from each individual family.
    \item \textbf{Curriculum learning}: We use Genet~\cite{xia2022genet} to iteratively generate increasingly challenging network scenarios and continue training the RL policy on the discovered traces, thereby exposing it to a progressively more challenging curriculum~\cite{akkaya2019solving,graves2017automated}. For all use cases, we run Genet for 20 rounds, doubling the 10 rounds used in the original paper.
\end{itemize}
Because Sage~\cite{yen2023sage} is an offline RL method, we fine-tune it following the procedure in Mowgli~\cite{agarwal2025mowgli}: we first collect trajectories of the RL policy on the discovered scenarios and then continue training the model offline on those trajectories (testbed details appear in Appendix~\ref{apdx:testbed}).

\subsection{Results}


\begin{graytxtframe}
\textbf{Finding 1:} \sys finds the largest performance gaps in all three use cases, certifying severe performance gaps of RL controllers.
\end{graytxtframe}

\sys discovers the largest avoidable performance gaps across all three controllers, outperforming the strongest discovery baselines by $1.6\times$ to $7.0\times$ (\ref{e1}, Fig.~\ref{fig:gap_plots}).
Specifically, \sys exposes a mean log QoE gap of $4.84$ in Pensieve ($7.0\times$ larger than Random on the original scale), a $63.8\%$ performance gap in Sage ($1.6\times$ larger than Random), and a $43.5\%$ completion-time gap in Park ($6.2\times$ larger than Gilad et al.), compared to baseline gaps of $1.09$, $1.7\%$, and $0.3\%$ on Normal traces.
In Park, the wider spread reflects workload heterogeneity across bursty job traces, where \sys consistently exposes the largest gap to Least Completion Time across all five workload distributions.

\smallskip
\begin{graytxtframe}
\textbf{Finding 2:} \sys, when derived from the largest-gap network scenarios it finds, provides the strongest protection among the compared methods in all three use cases, preserves nominal performance, and remains effective beyond its training scenario family.
\end{graytxtframe}

In its very first iteration, \sys removes $85.0\%$ of the normalized performance gap on Pensieve, $80.0\%$ on Sage, and $79.1\%$ on Park without retraining (\ref{e2}, Fig.~\ref{fig:heatmaps}).
Rules learned from \sys's scenarios also transfer across scenario families, removing $75.2\%$ of the gap on Random traces for Pensieve, $90.1\%$--$98.5\%$ on Gilad et al. and Genet traces for Sage, and $68.2\%$ on Gilad et al. traces for Park.
On Normal traces, \sys preserves and even reduces the residual nominal gap by $18.7\%$--$36.1\%$, whereas retraining baselines worsen the Normal gap by up to $24.3\%$ on Pensieve, $81.1\%$ on Park, and $15\times$ on Sage (Appendix~\ref{app:eval-details}).
Rules derived from Random scenarios achieve moderate transfer because RL policy flaws span contiguous regions of the state space~\cite{gilad_robustifying_2019,ray2022ccfuzz}, yet Random discovers much smaller gaps than \sys and yields weaker protection overall.

\smallskip
\begin{graytxtframe}
\textbf{Finding 3:} Learning protection rules from scenarios with larger performance gaps produces stronger runtime protection.
\end{graytxtframe}

Across all three controllers, \sys achieves higher average gap reduction when its rules are synthesized from scenarios with larger discovered regret (\ref{e3}, Fig.~\ref{fig:ablation}).
Rules learned from \sys's high-regret trajectories consistently outperform rules learned from weaker scenario generators such as Indago, Genet, or Random, confirming that stronger scenario discovery directly produces stronger runtime protection.

\smallskip
\begin{graytxtframe}
\textbf{Finding 4:} \sys always fits within the inference time budget.
\end{graytxtframe}

\sys's prefix-trie rule evaluation stays well below every controller's real-time decision deadline across all refinement rounds (\ref{e4}, Fig.~\ref{fig:controller_latency_plots}).
Across iterations 0--3, total decision latency consumes only $11.2\%$--$32.4\%$ of Pensieve's chunk-download budget, $47.6\%$--$50.1\%$ of Sage's $10\text{ ms}$ control interval (compared to $41.4\%$ for unmodified Sage), and $1.9\%$--$6.5\%$ of Park's inter-arrival budget.
Adding refinement rounds strengthens protection without increasing inference latency.

\smallskip
\begin{graytxtframe}
\textbf{Finding 5:} \sys does not merely protect against a specific challenging scenario; it fundamentally shrinks the worst-case performance gap.
\end{graytxtframe}

Running \sys's search-and-protect loop for three iterations shrinks the mean discovered gap by $74.4\%$ in Pensieve (from $4.84$ to $1.24$), $94.1\%$ in Sage (from $63.8\%$ to $3.8\%$), and $92.7\%$ in Park (from $33.5\%$ to $2.4\%$), closing $93.5\%$--$96.6\%$ of the distance to the Normal regime within two to three rounds (\ref{e5}, Fig.~\ref{fig:iterations}).

\section{Analyzing Revealed Performance Failures}
\label{sec:analysis}

Inspecting the scenarios and logic rules produced by \sys reveals the concrete policy flaws behind each controller's worst-case performance gap.
In Park, our analysis exposes a systematic dispatch flaw that appears directly in \sys's learned rules: when small jobs dominate the arrival stream, Park repeatedly routes them to slow servers even while fast servers sit idle.
We focus here on this Park mechanism and defer supporting statistics and the corresponding Pensieve and Sage analyses to Appendices~\ref{app:analysis-park-details}--\ref{app:analysis-sage}.

\begin{figure}[!t]
    \centering
    \begin{adjustbox}{width=1\linewidth,center=0pt}
    \includegraphics[width=\linewidth]{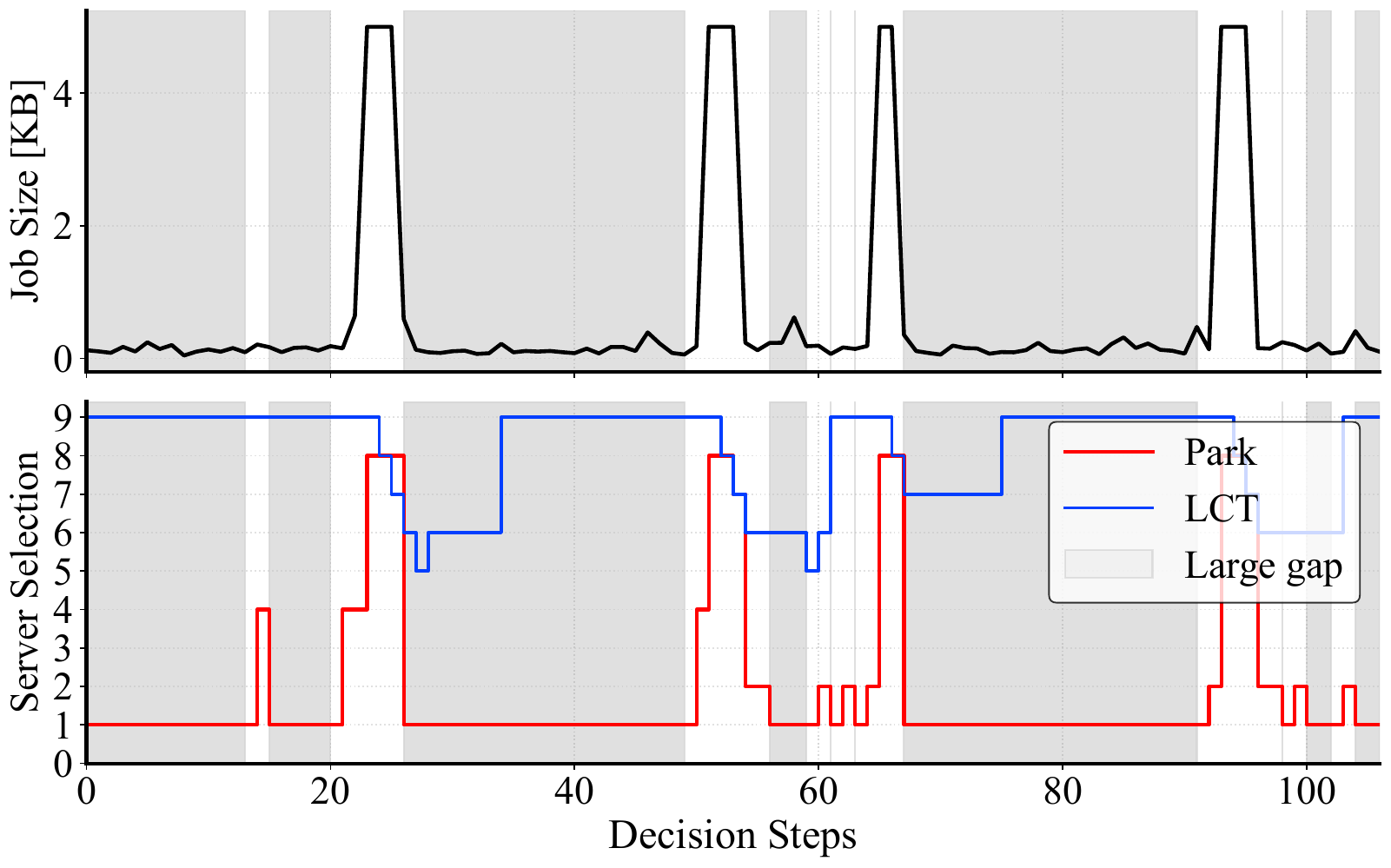}
    \end{adjustbox}
\caption{Small-job-heavy scenarios expose Park's bias toward slow servers.
Park keeps routing tiny jobs to the slow tier even while fast servers remain idle, which creates prolonged queueing and a large gap to Least Completion Time (LCT).}
\label{fig:park_failures1}
\end{figure}

\begin{figure}[!t]
    \centering
    \begin{adjustbox}{width=1.1\linewidth,center=0pt}
    \includegraphics[width=\linewidth]{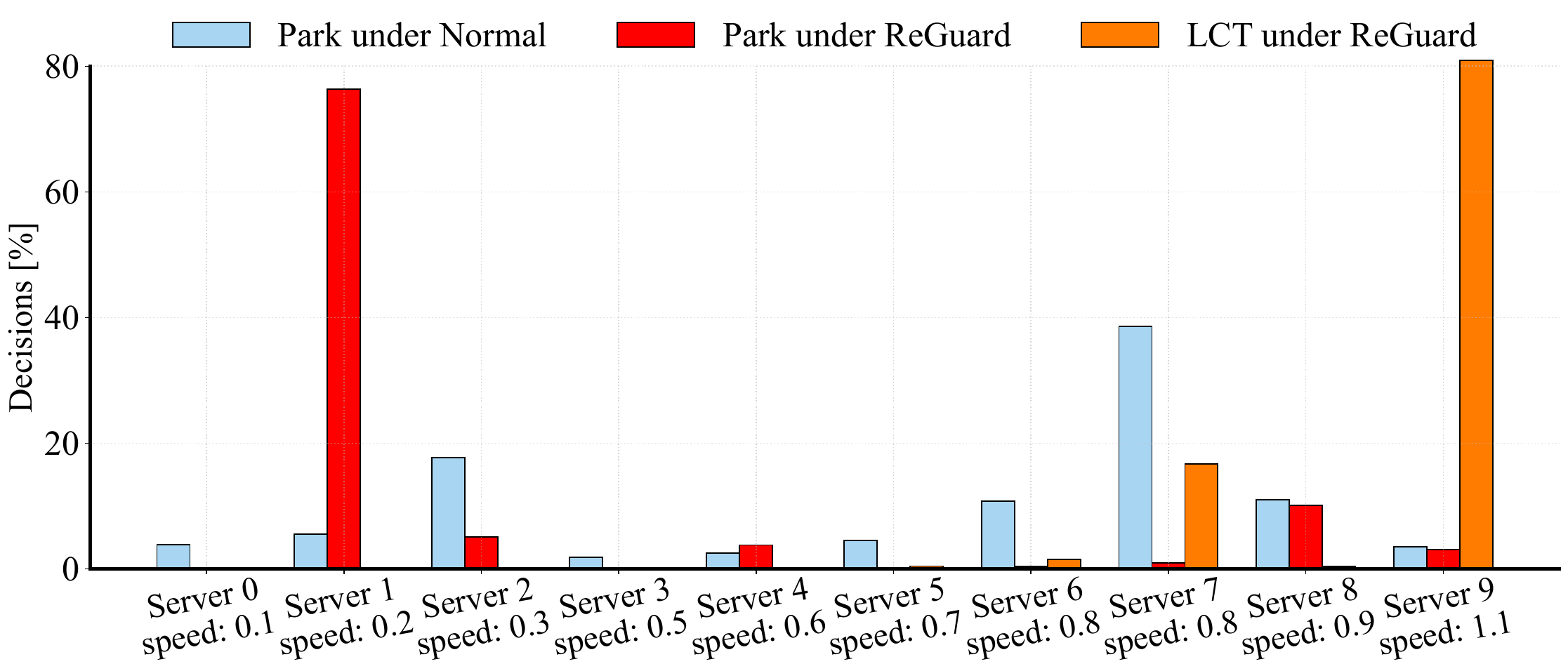}
    \end{adjustbox}
\caption{Network scenarios found by \sys turn Park's mild preference under nominal conditions into a near-collapse onto slow Server 1.
LCT still routes most jobs to the fast tier, which shows that the issue is Park's dispatch policy rather than an inherently bad workload.}
\label{fig:park_failures2}
\end{figure}

Park's degradation stems from an internal dispatch bias toward slow servers rather than cluster-wide overload.
As shown in Fig.~\ref{fig:park_failures1}, when short bursts of large jobs are followed by a stream of small jobs ($\le 0.2\text{ KB}$) that the fast tier could easily absorb, Park routes $97.3\%$ of those small jobs to slow servers and $86.5\%$ to slow Server 1 alone.
Under the exact same arrival stream, Least Completion Time (LCT) routes $81.1\%$ of those small jobs to idle fast servers.
Fig.~\ref{fig:park_failures2} confirms this pattern across the full workload corpus: under \sys's scenarios, Park sends $76.3\%$ of all jobs to slow Server 1 (up from $5.6\%$ on normal workloads) even though at least one fast server sits completely idle on $98.0\%$ of decision steps.

\sys makes this dispatch flaw explicit by separating risk detection from server-mask recovery using percentile predicates $\mathrm{p}k$ calibrated on normal traces (\S\ref{sec:shield}).
In the rules below, $\mathrm{GapToMin}_i$ measures how much server $i$'s current load exceeds the least-loaded server, $\mathrm{PostLoad}_i$ denotes server $i$'s projected load after adding the arriving job, and $\mathrm{MeanLoad}$, $\mathrm{StdLoad}$, and $\mathrm{TotalLoad}$ summarize cluster-wide queue state.
First, a representative risk rule flags incoming small jobs whenever slow Server 1 is backed up while fast Server 9 has spare capacity:
\begin{equation}
\label{eq:park_risk_rule}
\begin{aligned}
&(\mathrm{JobSize} \le \mathrm{p25}) \land (\mathrm{GapToMin}_1 > \mathrm{p25}) \\
&{}\land (\mathrm{PostLoad}_9 \le \mathrm{p25}) \land (\mathrm{StdLoad} > \mathrm{p25}) \\
&{}\implies \mathrm{Risky}=\texttt{true}.
\end{aligned}
\end{equation}
Equation~\eqref{eq:park_risk_rule} captures the exact failure regime in Fig.~\ref{fig:park_failures1}: an easy job arrives ($\mathrm{JobSize} \le \mathrm{p25}$), load is unevenly distributed across the cluster ($\mathrm{StdLoad} > \mathrm{p25}$) with slow Server 1 carrying excess work ($\mathrm{GapToMin}_1 > \mathrm{p25}$), yet fast Server 9 would remain lightly loaded even after taking the job ($\mathrm{PostLoad}_9 \le \mathrm{p25}$).

When $\mathrm{Risky}=\texttt{true}$, \sys evaluates recovery rules that construct the safe server mask $\mathcal{A}_{\mathrm{safe}}(s_t)$ by reopening underutilized fast servers:
\begin{flalign}
\label{eq:park_allow9}
&\begin{aligned}
&(\mathrm{GapToMin}_1 > \mathrm{p25}) \land (\mathrm{GapToMin}_9 \le \mathrm{p25}) \\
&{}\quad\land (\mathrm{MeanLoad} \le \mathrm{p25}) \land (\mathrm{TotalLoad} \le \mathrm{p25}) \\
&{}\quad\implies \mathrm{Allow}_9=\texttt{true},
\end{aligned}&&
\end{flalign}
\begin{flalign}
\label{eq:park_allow7}
&\begin{aligned}
&(\mathrm{Load}_9 \le \mathrm{p25}) \land (\mathrm{PostLoad}_2 \le \mathrm{p25}) \\
&{}\quad\land (\mathrm{MeanLoad} \le \mathrm{p25}) \land (\mathrm{StdLoad} \le \mathrm{p25}) \\
&{}\quad\implies \mathrm{Allow}_7=\texttt{true}.
\end{aligned}&&
\end{flalign}
Equations~\eqref{eq:park_allow9} and~\eqref{eq:park_allow7} show that \sys reopens multiple fast servers (Servers 7 and 9) whenever the cluster is lightly loaded overall but slow Server 1 is disproportionately busy, preventing Park from collapsing onto the slow tier while preserving Park's relative ranking among the allowed fast servers (Appendix~\ref{app:analysis-park-details}).

\section{Conclusion}
\sys discovers worst-case failures in RL-based network controllers by formulating the search for challenging network scenarios as bilevel regret maximization.
It protects the controller at inference time using lightweight, interpretable logic rules derived from counterfactual analysis.
Across three controllers spanning distinct networking tasks, \sys exposes performance gaps up to 6$\times$ larger than prior methods and closes up to 85\% of those gaps within the first refinement iteration, without retraining and without sacrificing nominal performance.


\section*{Acknowledgments}

We thank the anonymous reviewers for their insightful feedback and comments, which meaningfully improved this paper.
This work was supported by the National Science Foundation (NSF) through Grants CNS-2442625 and CNS-2319442 and Office of Naval Research (ONR) through Grant N000142612552.

\clearpage


\bibliographystyle{plainnat}
\bibliography{reference}

\newpage
\appendix
\section{Proofs for Regret Guarantees}
\label{app:proofs}

This appendix gives the detailed assumptions, intermediate lemmas, and full proofs for the regret guarantees in Section~\ref{sec:guarantees_sys}.
Our goal is to make explicit what is guaranteed by the formulation, what approximation errors enter the final bound, and how those errors propagate through the argument.

\subsection{Setup and Notation}
\label{app:proof_setup}

We begin by restating the exact and approximate regret objectives used in the main text.
For every feasible network scenario $e \in \mathcal{E}$, define
\begin{equation}
\label{eq:app_regret_defs}
R(e) := J(\pi_e^\star;e) - J(\pi;e),
\qquad
\hat{R}(e) := J(\hat{\pi}_e^\star;e) - J(\pi;e).
\end{equation}
Here:
\begin{itemize}[leftmargin=*, itemsep=0pt, topsep=2pt]
    \item $R(e)$ is the \emph{exact} regret objective, where the comparator $\pi_e^\star$ is an exact optimizer within the reference policy class:
    \[
    \pi_e^\star \in \arg\max_{\pi' \in \Pi} J(\pi';e).
    \]
    \item $\hat{R}(e)$ is the \emph{approximate} regret objective, where the comparator $\hat{\pi}_e^\star \in \Pi$ is the practical reference policy used by the implementation.
\end{itemize}

The exact bilevel target in the main text is
\begin{equation}
\label{eq:app_outer_true}
e^\star \in \arg\max_{e \in \mathcal{E}} R(e),
\end{equation}
while the implemented procedure approximately solves
\begin{equation}
\label{eq:app_outer_hat}
\hat{e} \in \arg\max_{e \in \mathcal{E}} \hat{R}(e).
\end{equation}

The proofs below establish three facts.
First, the outer solver returns a near-optimal scenario for the approximate objective.
Second, the approximate objective is a pointwise lower bound on the exact objective.
Third, combining these two facts yields the final exact-regret guarantee in Theorem~\ref{thm:exact}.

\subsection{Assumptions}
\label{app:assumptions}

We restate the assumptions from the main text for completeness.

\begin{enumerate}[leftmargin=*, label=A\arabic*:, ref=A\arabic*, itemsep=0pt, topsep=2pt]
    \item \label{as:nonempty}
    The feasible set $\mathcal{E}$ is non-empty, and both Eqns.~\eqref{eq:outer_true} and~\eqref{eq:outer_hat} admit optimizers.

    \item \label{as:finite}
    For every $e \in \mathcal{E}$, the rewards $J(\pi_e^\star;e)$, $J(\hat{\pi}_e^\star;e)$, and $J(\pi;e)$ are finite.

    \item \label{as:solver}
    The outer RL solver returns an $\varepsilon$-optimal solution $\tilde e \in \mathcal{E}$ to Eqn.~\eqref{eq:outer_hat}, namely
    \begin{equation}
    \label{eq:ass_solver}
    \max_{e \in \mathcal{E}} \hat{R}(e) - \hat{R}(\tilde e) \le \varepsilon.
    \end{equation}

    \item \label{as:oracle}
    The approximate reference satisfies
    \begin{equation}
    \label{eq:ass_oracle}
    J(\pi_e^\star;e) - J(\hat{\pi}_e^\star;e) \le \delta(e)
    \end{equation}
    for every $e \in \mathcal{E}$, where $\delta(e) \ge 0$.
\end{enumerate}

These assumptions are mild in the settings studied in the paper.
\ref{as:nonempty} says that the operator-defined feasible family of network scenarios is nontrivial and that both optimization problems are well posed.
\ref{as:finite} excludes degenerate cases in which the trajectory reward is undefined or unbounded.
\ref{as:solver} captures the quality of the outer solver.
It does not require exact optimization, only that the scenario returned by the solver be within $\varepsilon$ of the best achievable approximate-regret score.
\ref{as:oracle} captures the quality of the inner reference policy.
It says that the approximate reference does not fall more than $\delta(e)$ below the exact best reference in scenario $e$.

The final theorem separates these two error sources cleanly.
The term $\varepsilon$ accounts for suboptimality in the outer search over scenarios.
The term $\delta(e^\star)$ accounts for the error of the practical reference at an exact worst-case scenario.
This separation is useful because the two terms arise from different components of the system and can be improved independently.

\subsection{Auxiliary Observations}
\label{app:aux_obs}

We first record two elementary but useful observations that will be used repeatedly in the proofs.

\begin{lem}
\label{lem:exact_ge_approx_ref}
For every feasible scenario $e \in \mathcal{E}$,
\begin{equation}
\label{eq:exact_ge_approx_ref}
J(\pi_e^\star;e) \ge J(\hat{\pi}_e^\star;e).
\end{equation}
\end{lem}

\begin{proof}
Fix any $e \in \mathcal{E}$.
By definition, $\pi_e^\star$ is an optimizer of $J(\pi';e)$ over the reference class $\Pi$.
Since $\hat{\pi}_e^\star \in \Pi$, the optimality of $\pi_e^\star$ implies
\[
J(\pi_e^\star;e) \ge J(\hat{\pi}_e^\star;e).
\]
This proves the claim.
\end{proof}

\begin{lem}
\label{lem:regret_diff_identity}
For every feasible scenario $e \in \mathcal{E}$,
\begin{equation}
\label{eq:regret_diff_identity}
R(e) - \hat{R}(e)
=
J(\pi_e^\star;e) - J(\hat{\pi}_e^\star;e).
\end{equation}
\end{lem}

\begin{proof}
By the definitions in Eqn.~\eqref{eq:app_regret_defs},
\[
R(e) = J(\pi_e^\star;e) - J(\pi;e),
\qquad
\hat{R}(e) = J(\hat{\pi}_e^\star;e) - J(\pi;e).
\]
Subtracting the second identity from the first cancels the common term $J(\pi;e)$ and yields
\[
R(e) - \hat{R}(e)
=
J(\pi_e^\star;e) - J(\hat{\pi}_e^\star;e).
\]
\end{proof}

The next observation makes explicit that the approximate objective is always conservative.

\begin{lem}
\label{lem:approx_is_conservative}
For every feasible scenario $e \in \mathcal{E}$,
\begin{equation}
\label{eq:approx_is_conservative}
\hat{R}(e) \le R(e).
\end{equation}
\end{lem}

\begin{proof}
Fix any $e \in \mathcal{E}$.
By Lemma~\ref{lem:exact_ge_approx_ref},
\[
J(\pi_e^\star;e) \ge J(\hat{\pi}_e^\star;e).
\]
Subtracting the common term $J(\pi;e)$ from both sides preserves the inequality, giving
\[
J(\pi_e^\star;e)-J(\pi;e)
\ge
J(\hat{\pi}_e^\star;e)-J(\pi;e).
\]
By Eqn.~\eqref{eq:app_regret_defs}, this is exactly
\[
R(e) \ge \hat{R}(e).
\]
\end{proof}

Lemma~\ref{lem:approx_is_conservative} is conceptually important.
It says that replacing the exact reference policy with the practical reference can only decrease the regret score.
In other words, the approximate regret objective used by the implementation is a pointwise lower bound on the ideal objective of interest.

\subsection{Outer-search Guarantee}
\label{app:proof-outer}

We now formalize the guarantee provided by the outer RL solver on the approximate objective.

\begin{lem}
\label{thm:outer}
Under \ref{as:nonempty} and~\ref{as:solver}, the scenario $\tilde e$ returned by the outer solver satisfies
\begin{equation}
\label{eq:guar1}
\max_{e \in \mathcal{E}} \hat{R}(e) - \hat{R}(\tilde e) \le \varepsilon.
\end{equation}
Equivalently,
\begin{equation}
\label{eq:app_guar1_alt}
\hat{R}(\tilde e) \ge \max_{e \in \mathcal{E}} \hat{R}(e) - \varepsilon.
\end{equation}
\end{lem}

\begin{proof}
By \ref{as:nonempty}, the approximate outer problem in Eqn.~\eqref{eq:outer_hat} admits an optimizer.
Let $\hat e^\star \in \mathcal{E}$ be such that
\begin{equation}
\label{eq:hat_opt_attained}
\hat{R}(\hat e^\star) = \max_{e \in \mathcal{E}} \hat{R}(e).
\end{equation}
By \ref{as:solver}, the outer solver returns a scenario $\tilde e \in \mathcal{E}$ whose approximate objective value is within $\varepsilon$ of the optimum:
\begin{equation}
\label{eq:outer_solver_eps_again}
\max_{e \in \mathcal{E}} \hat{R}(e) - \hat{R}(\tilde e) \le \varepsilon.
\end{equation}
This is exactly Eqn.~\eqref{eq:guar1}.
Rearranging Eqn.~\eqref{eq:guar1} gives
\[
\hat{R}(\tilde e)
\ge
\max_{e \in \mathcal{E}} \hat{R}(e)-\varepsilon,
\]
which is Eqn.~\eqref{eq:app_guar1_alt}.
\end{proof}

Although straightforward, Lemma~\ref{thm:outer} plays an important role.
It tells us that the outer solver returns a scenario whose \emph{approximate} regret score is nearly the largest possible over the feasible family.
This is the first half of the final argument.
The second half is to connect approximate regret to exact regret.

\subsection{Pointwise Approximation Guarantee}
\label{app:proof-pointwise}

We next show that the approximate regret is a pointwise lower bound on exact regret and that the gap is controlled exactly by the inner-reference error.

\begin{lem}
\label{thm:pointwise}
Under \ref{as:finite} and~\ref{as:oracle}, every $e \in \mathcal{E}$ satisfies
\begin{equation}
\label{eq:app_pointwise_gap}
0 \le R(e) - \hat{R}(e) \le \delta(e).
\end{equation}
\end{lem}

\begin{proof}
Fix any $e \in \mathcal{E}$.
We prove the lower and upper bounds separately.

\mypar{Lower bound}
By Lemma~\ref{lem:regret_diff_identity},
\[
R(e)-\hat{R}(e)
=
J(\pi_e^\star;e)-J(\hat{\pi}_e^\star;e).
\]
By Lemma~\ref{lem:exact_ge_approx_ref},
\[
J(\pi_e^\star;e)-J(\hat{\pi}_e^\star;e) \ge 0.
\]
Combining the two equations yields
\begin{equation}
\label{eq:pointwise_lower_detailed}
R(e)-\hat{R}(e) \ge 0.
\end{equation}

\mypar{Upper bound}
Again by Lemma~\ref{lem:regret_diff_identity},
\[
R(e)-\hat{R}(e)
=
J(\pi_e^\star;e)-J(\hat{\pi}_e^\star;e).
\]
\ref{as:oracle} states exactly that
\[
J(\pi_e^\star;e)-J(\hat{\pi}_e^\star;e) \le \delta(e).
\]
Therefore,
\begin{equation}
\label{eq:pointwise_upper_detailed}
R(e)-\hat{R}(e) \le \delta(e).
\end{equation}

Combining Eqns.~\eqref{eq:pointwise_lower_detailed} and~\eqref{eq:pointwise_upper_detailed} proves
\[
0 \le R(e)-\hat{R}(e) \le \delta(e).
\]
\end{proof}

Lemma~\ref{thm:pointwise} is the precise statement that the approximate objective is conservative.
At every scenario $e$, the exact regret is at least as large as the approximate regret.
Furthermore, the amount by which the approximate objective can underestimate the exact objective is no larger than the inner-reference error $\delta(e)$.
This pointwise control is what allows us to lift an approximate-optimality guarantee for the outer solver into an exact-regret certificate.

\subsection{Exact-regret Guarantee}
\label{app:proof-exact}

We now prove the main guarantee from the paper.

\begin{proof}[Proof of Theorem~\ref{thm:exact}]
By \ref{as:nonempty}, the exact outer problem in Eqn.~\eqref{eq:outer_true} admits an optimizer.
Let
\begin{equation}
\label{eq:choose_exact_opt}
e^\star \in \arg\max_{e \in \mathcal{E}} R(e).
\end{equation}
Then by definition,
\begin{equation}
\label{eq:estar_exact_opt}
R(e^\star)=\max_{e \in \mathcal{E}} R(e).
\end{equation}

Our goal is to lower-bound $R(\tilde e)$ in terms of $R(e^\star)$ and then substitute Eqn.~\eqref{eq:estar_exact_opt}.

\mypar{Step 1: compare the returned scenario to the best approximate-regret scenario}
By Lemma~\ref{thm:outer},
\begin{equation}
\label{eq:outer_used_exact}
\hat{R}(\tilde e)
\ge
\max_{e \in \mathcal{E}} \hat{R}(e)-\varepsilon.
\end{equation}
Since $e^\star \in \mathcal{E}$, the maximum over $\mathcal{E}$ is at least the value at $e^\star$.
Therefore,
\begin{equation}
\label{eq:compare_with_e_star}
\hat{R}(\tilde e)
\ge
\hat{R}(e^\star)-\varepsilon.
\end{equation}

\mypar{Step 2: relate exact and approximate regret at the returned scenario}
Applying Lemma~\ref{thm:pointwise} to the scenario $\tilde e$ gives
\begin{equation}
\label{eq:tilde_exact_ge_hat}
R(\tilde e) \ge \hat{R}(\tilde e).
\end{equation}

\mypar{Step 3: relate exact and approximate regret at the exact worst-case scenario}
Applying Lemma~\ref{thm:pointwise} to the scenario $e^\star$ gives
\[
0 \le R(e^\star)-\hat{R}(e^\star) \le \delta(e^\star).
\]
Rearranging the upper bound yields
\begin{equation}
\label{eq:e_star_hat_lower}
\hat{R}(e^\star) \ge R(e^\star)-\delta(e^\star).
\end{equation}

\mypar{Step 4: chain the inequalities}
Starting from Eqn.~\eqref{eq:tilde_exact_ge_hat} and then using Eqns.~\eqref{eq:compare_with_e_star} and~\eqref{eq:e_star_hat_lower}, we obtain
\begin{align}
R(\tilde e)
&\ge \hat{R}(\tilde e)
\label{eq:chain1}
\\
&\ge \hat{R}(e^\star)-\varepsilon
\label{eq:chain2}
\\
&\ge R(e^\star)-\delta(e^\star)-\varepsilon.
\label{eq:chain3}
\end{align}
Substituting Eqn.~\eqref{eq:estar_exact_opt} into Eqn.~\eqref{eq:chain3} gives
\begin{equation}
\label{eq:exact_main_lb}
R(\tilde e)
\ge
\max_{e \in \mathcal{E}} R(e)-\varepsilon-\delta(e^\star).
\end{equation}
This is exactly Eqn.~\eqref{eq:true_gap_alt} in the main text.

Finally, rearranging Eqn.~\eqref{eq:exact_main_lb} yields
\begin{equation}
\label{eq:exact_main_gap}
\max_{e \in \mathcal{E}} R(e)-R(\tilde e)
\le
\varepsilon+\delta(e^\star),
\end{equation}
which is equivalent to Eqn.~\eqref{eq:true_gap_alt} in the main text.
This completes the proof.
\end{proof}

\subsection{Interpretation of the Final Bound}
\label{app:interpretation_final_bound}

Theorem~\ref{thm:exact} has a clean interpretation.

First, the regret achieved by the returned scenario, $R(\tilde e)$, is itself a \emph{certificate value}.
It lower-bounds the worst-case exact regret over the feasible family up to the additive error $\varepsilon+\delta(e^\star)$.
Thus, if $R(\tilde e)$ is large, then the controller provably has a substantial avoidable performance gap in the searched regime.

Second, the theorem isolates the only two sources of looseness in the certificate:
\begin{itemize}[leftmargin=*, itemsep=0pt, topsep=2pt]
    \item $\varepsilon$ is the suboptimality of the outer search over scenarios.
    It is controlled by the quality of the RL-based outer solver.
    \item $\delta(e^\star)$ is the suboptimality of the practical reference at an exact worst-case scenario.
    It is controlled by the quality of the inner reference policy.
\end{itemize}

Third, the dependence on $\delta(e^\star)$ rather than a worst-case uniform bound over all $e \in \mathcal{E}$ is meaningful.
The theorem only needs the inner-reference quality at an exact optimizer of the true objective.
If the practical reference is especially accurate near the most informative failure scenarios, then the final certificate can still be tight even if the reference is less accurate elsewhere.

\subsection{A Uniform-error Corollary}
\label{app:uniform_cor}

In some settings, it is convenient to summarize inner-reference quality by a uniform bound.
The next corollary makes this explicit.

\begin{cor}
\label{cor:uniform_delta}
Suppose the assumptions of Theorem~\ref{thm:exact} hold, and in addition there exists a constant $\bar{\delta} \ge 0$ such that
\begin{equation}
\label{eq:uniform_delta}
\delta(e) \le \bar{\delta},
\qquad
\forall e \in \mathcal{E}.
\end{equation}
Then the returned scenario $\tilde e$ satisfies
\begin{equation}
\label{eq:uniform_delta_bound}
\max_{e \in \mathcal{E}} R(e)-R(\tilde e)
\le
\varepsilon+\bar{\delta},
\end{equation}
or equivalently,
\begin{equation}
\label{eq:uniform_delta_bound_alt}
R(\tilde e)
\ge
\max_{e \in \mathcal{E}} R(e)-\varepsilon-\bar{\delta}.
\end{equation}
\end{cor}

\begin{proof}
By Theorem~\ref{thm:exact},
\[
\max_{e \in \mathcal{E}} R(e)-R(\tilde e)
\le
\varepsilon+\delta(e^\star).
\]
Under Eqn.~\eqref{eq:uniform_delta}, we have $\delta(e^\star)\le \bar{\delta}$.
Substituting this into the previous inequality yields Eqn.~\eqref{eq:uniform_delta_bound}.
Rearranging gives Eqn.~\eqref{eq:uniform_delta_bound_alt}.
\end{proof}

Corollary~\ref{cor:uniform_delta} is weaker than Theorem~\ref{thm:exact}, but sometimes easier to state.
It says that if the practical reference is uniformly within $\bar{\delta}$ of the exact reference everywhere, then the final exact-regret certificate loses at most $\varepsilon+\bar{\delta}$ compared to the ideal exact worst case.

\subsection{What the Guarantees Do and Do Not Say}
\label{app:scope_guarantees}

For completeness, we summarize the scope of the above results.

\mypar{What the guarantees establish}
The proofs show that, under \ref{as:nonempty}--\ref{as:oracle}, the returned scenario $\tilde e$ certifies a large avoidable performance gap for the pretrained controller.
More precisely, the exact regret value $R(\tilde e)$ is guaranteed to be close to the largest exact regret achievable over the feasible search space.
The approximation error is explicit and decomposes additively into outer-search error and inner-reference error.

\mypar{What the guarantees do not establish}
The guarantees do not claim that the returned scenario is unique.
They also do not claim that a particular nonconvex RL algorithm always meets \ref{as:solver}.
Instead, the theorem is conditional: \emph{if} the outer solver is $\varepsilon$-optimal for the implemented approximate objective, and \emph{if} the practical reference is within $\delta(e)$ of the exact reference, \emph{then} the returned scenario is a near-tight certificate for the exact worst-case regret over the feasible family.

\mypar{Why the conditional form is still useful}
This conditional structure matches how systems are analyzed in practice.
The theoretical formulation makes precise what quantity the method is targeting.
The theorem then shows exactly how optimization and approximation errors degrade the final certificate.
As the outer solver improves or the practical reference becomes stronger, the guarantee strengthens immediately and transparently.

\section{Use Cases}
\label{app:use-cases}

This appendix summarizes the three controller settings used throughout the paper.

\mypar{Adaptive bitrate streaming}
An adaptive bitrate controller chooses the bitrate of the next video chunk at each chunk boundary.
Its decision is based on the recent throughput history, the current playback buffer, and information about upcoming chunks.
Pensieve~\cite{mao_neural_2017_pensieve} is an RL-based instance of this setting.
Its objective is to keep video quality high over the full session while limiting rebuffering and avoiding abrupt bitrate swings.

\mypar{Congestion control}
A congestion-control sender adjusts its transmission behavior from transport feedback observed on the path, such as delay, loss, acknowledgments, and delivery rate.
An RL-based policy in this setting makes a new control decision at the beginning of each interval, whose duration is tied to the path dynamics.
Sage~\cite{yen2023sage} is the RL-based congestion-control system studied in this paper.
Its reward combines throughput, latency, and packet loss, so the controller must push traffic efficiently without driving the path into persistent queueing or loss.
Router-side mechanisms address a complementary control surface: Titrate adapts queue thresholds to balance throughput, latency, and burst absorption~\cite{zhou2026titrate}.

\mypar{Load balancing}
A load balancer in a replicated distributed store assigns each arriving request to one of several candidate servers.
The balancer observes request-arrival patterns, recent request sizes, and its own estimate of outstanding work already sent to each server, but it does not directly observe each server's instantaneous internal utilization.
Park~\cite{mao2019park} is the RL-based policy we study in this setting.
Its goal is to route requests so that system-wide service remains efficient despite skewed arrivals, heterogeneous service speeds, and incomplete visibility into the servers' real-time state.

\section{Additional Implementation Details}
\label{app:additional-implementations}

\subsection{Bilevel Search Implementations for Pensieve and Park}
\label{app:maxregret-implementations}

\mypar{Adaptive bitrate streaming}
Pensieve~\cite{mao_neural_2017_pensieve} instantiates the same formulation in a chunk-level ABR simulator.
Here the unique design choice is not process synchronization, but an exact short-horizon reference oracle.
At each chunk, $u_t$ is the bandwidth used to download the next video chunk, and the simulator updates download time, buffer occupancy, rebuffering, sleep behavior, and QoE through Pensieve's standard dynamics.
The solver controls only this bandwidth process.
Pensieve still chooses bitrates through its normal observation interface, so the generated trace affects the controller only through measured throughput, delay, buffer state, next chunk sizes, and remaining chunks.

\sys applies Eqn.~\eqref{eq:impl_reference} exactly over a rolling window of $K$ chunks.
For the realized bandwidths and the buffer state at the beginning of the window, \sys enumerates candidate bitrate sequences in $\Pi$ and selects the sequence with the largest QoE reward from Table~\ref{tab:maxregret_impl}.
With six bitrate levels, this local search has size $6^K$, which is tractable for small $K$ and avoids the infeasible full-session search over $6^{48}$ sequences.
This rolling oracle gives the outer solver dense feedback about avoidable ABR mistakes, such as over-aggressive bitrate increases before a bandwidth drop or conservative decisions that miss short high-bandwidth opportunities.

\mypar{Load balancing}
Park~\cite{mao2019park} instantiates the formulation as closed-loop workload generation for load balancing.
At each load-balancing decision, $u_t$ specifies the next inter-arrival time and job size, and the simulator advances queues, running jobs, completion events, and active-job-time reward after Park dispatches the current job.
The action variables are typically log-scaled because both arrival gaps and job sizes span orders of magnitude.
This lets the solver express bursts, idle periods, small jobs, and large jobs without leaving the controlled workload family.

Park uses a set of traditional heuristics as its portfolio for $\Pi$: Least Completion Time, Join Shortest Queue, and Choose Fastest Server.
\sys applies Eqn.~\eqref{eq:impl_reference} by rolling each heuristic forward from the same simulator snapshot over the same generated workload window and selecting the best active-job-time reward.
The distinctive implementation issue is snapshot-consistent counterfactual evaluation.
The arrival stream, server speeds, current incoming job, and initial simulator state must be identical across policies, while each policy must have independent queues, running jobs, event timelines, and reward accounting.
Deep snapshots enforce this separation, so a positive gap means that Park is dominated by simple decision rules under the same workload rather than by an artifact of shared mutable state.
The resulting scenarios emphasize decision-sensitive workload patterns, such as bursts that make queue-length estimates misleading or size mixtures where a poor server choice creates head-of-line blocking.

\subsection{Protection Implementations for Pensieve and Park}
\label{app:shield-implementations}

\mypar{Adaptive bitrate streaming}
For Pensieve~\cite{mao_neural_2017_pensieve}, the protected controller consumes the $6 \times 8$ ABR state summarized in Table~\ref{tab:pensieve_state_features}.
This state contains histories of bitrate, buffer, throughput, delay, and chunks remaining, plus the next-chunk size vector for the available bitrate choices.
Pensieve's action space is discrete because each action selects one bitrate level.
For this reason, \sys implements the generic adjustment as a learned safe action interval rather than as a continuous nudge.
The rules infer lower and upper bitrate bounds $L(s_t)$ and $U(s_t)$, which define
$
\mathcal{A}_{\mathrm{safe}}(s_t)
:=
\{a \in \mathcal{A}: L(s_t) \le \mathrm{bitrate}(a) \le U(s_t)\}.
$
If Pensieve's chosen bitrate lies inside this set, \sys abstains.
If Pensieve chooses a bitrate above $U(s_t)$, \sys backs off by selecting Pensieve's highest-scored action inside $\mathcal{A}_{\mathrm{safe}}(s_t)$.
If Pensieve chooses a bitrate below $L(s_t)$, \sys pushes harder by selecting Pensieve's highest-scored action inside the safe set.
Thus, \sys masks unsafe bitrate choices while preserving Pensieve's ranking among the allowed actions:
$$
a_t^{\mathrm{prot}}
:=
\begin{cases}
a_t, & a_t \in \mathcal{A}_{\mathrm{safe}}(s_t),\\
\arg\max_{a \in \mathcal{A}_{\mathrm{safe}}(s_t)} p_t(a), & a_t \notin \mathcal{A}_{\mathrm{safe}}(s_t),
\end{cases}
$$
where $p_t(a)$ is Pensieve's action score or probability for bitrate action $a$.

\mypar{Load balancing}
For Park~\cite{mao2019park}, the protected controller observes the load-proxy state summarized in Table~\ref{tab:park_state_features} before dispatching an incoming job.
The original policy input contains only per-server load proxies and the incoming job size.
For rule learning, \sys derives summaries such as total load, mean load, load imbalance, server load ranks, and each server's load after hypothetically adding the incoming job from the same observation.
These derived features are not additional inputs to the original Park policy.
Park's action space is discrete because each action selects a server.
\sys implements protection as a risk rule followed by a per-server allow mask.
The learned rules define a safe action set
$
\mathcal{A}_{\mathrm{safe}}(s_t)
:=
\{i \in \mathcal{A}: \mathrm{AllowServer}_i(s_t)=1\}.
$
If no risk rule matches, all servers are treated as safe and \sys abstains.
If Park's selected server is in $\mathcal{A}_{\mathrm{safe}}(s_t)$, the action is left unchanged.
If the state is risky and Park's selected server is outside the safe set, \sys selects the highest-scored Park action among the allowed servers:
$$
a_t^{\mathrm{prot}}
:=
\begin{cases}
a_t, &
\substack{
\mathrm{Risky}(s_t)=0\\
\text{or } a_t \in \mathcal{A}_{\mathrm{safe}}(s_t)
},\\
\arg\max_{i \in \mathcal{A}_{\mathrm{safe}}(s_t)} p_t(i), &
\substack{
\mathrm{Risky}(s_t)=1\\
\text{and } a_t \notin \mathcal{A}_{\mathrm{safe}}(s_t)
}.
\end{cases}
$$
This mapping realizes \textsc{back\_off} and \textsc{push\_harder} as application-specific server redirections rather than scalar action changes.
For example, a rule may prevent Park from sending another small job to an overloaded slow server and redirect the decision to the best Park-scored server among those supported by reference heuristics.
\sys therefore constrains only risky dispatches while preserving Park's preferences whenever they are compatible with the learned safe set.

\section{Controller State Features}
\label{app:state-features}

This appendix summarizes the controller state features referenced by the protection implementation in Section~\ref{sec:shield}.

\begin{table*}[t]
\centering
\footnotesize
\begin{adjustbox}{width=\textwidth,center}
\begin{tabularx}{\textwidth}{>{\raggedright\arraybackslash}p{0.10\textwidth}>{\raggedright\arraybackslash}p{0.09\textwidth}>{\raggedright\arraybackslash}X}
\toprule
Obs. index & Raw col. & Feature summary \\
\midrule
0 & 2 & Current smoothed RTT, normalized as RTT divided by 100 ms. \\
1 & 3 & Current RTT variation in ms. \\
2 & 7 & Current delivery rate, normalized by bandwidth factor. \\
3 & 9 & TCP congestion-avoidance state. \\
4--12 & 10--18 & RTT rolling summaries over short, medium, and long windows: average, minimum, and maximum. \\
13--21 & 19--27 & Delivery-rate rolling summaries over short, medium, and long windows: average, minimum, and maximum. \\
22--30 & 28--36 & Min-RTT-ratio and RTT-rate rolling summaries over short, medium, and long windows: average, minimum, and maximum. \\
31--39 & 37--45 & RTT-variation rolling summaries over short, medium, and long windows: average, minimum, and maximum. \\
40--48 & 46--54 & Inflight and unacked-data rolling summaries over short, medium, and long windows: average, minimum, and maximum. \\
49--57 & 55--63 & Lost-packet rolling summaries over short, medium, and long windows: average, minimum, and maximum. \\
58 & 65 & Time delta since the previous observation. \\
59 & 66 & Current min-RTT-ratio or RTT-rate signal. \\
60 & 67 & Current loss signal, normalized by bandwidth factor. \\
61 & 68 & ACKed rate. \\
62 & 69 & Delivery-rate growth ratio relative to the previous window. \\
63 & 70 & Max-delivery-rate or cwnd-style utilization ratio. \\
64 & 71 & Current windowed delivery rate, normalized by bandwidth factor. \\
65 & 72 & cwnd-unacked ratio. \\
66 & 73 & Max delivery-rate growth ratio relative to the previous max. \\
67 & 74 & Max recent windowed delivery rate, normalized by bandwidth factor. \\
68 & 76 & Previous Sage action or action feature. \\
\bottomrule
\end{tabularx}
\end{adjustbox}
\caption{Sage policy observation features used by \sys.
Sage consumes 69 selected entries from a larger 77-field shared-memory telemetry message.
The raw fields excluded from the policy observation are timestamp, path or bandwidth scalar, RTO, ATO, pacing rate, slow-start threshold, one net-goodput field, and reward.
In particular, reward is present in the raw shared-memory message but excluded from the policy observation.}
\label{tab:sage_state_features}
\end{table*}

\begin{table*}[t]
\centering
\footnotesize
\begin{adjustbox}{width=\textwidth,center}
\begin{tabularx}{\textwidth}{>{\raggedright\arraybackslash}p{0.10\textwidth}>{\raggedright\arraybackslash}p{0.20\textwidth}>{\raggedright\arraybackslash}p{0.42\textwidth}>{\raggedright\arraybackslash}X}
\toprule
State row & Feature & Meaning & Normalization \\
\midrule
0 & Last selected bitrate & Bitrate selected for the most recently downloaded chunk. & $s_0=b_t/\max_b$ \\
1 & Buffer occupancy & Current playback buffer size. & $s_1=B_t/10$ \\
2 & Measured throughput & Throughput observed from the last chunk download. & $s_2=S_t/(d_t\cdot 1000)$ \\
3 & Download delay & Time taken to download the last chunk. & $s_3=d_t/(1000\cdot 10)$ \\
4 & Next chunk sizes & Byte sizes of the next video chunk at each available bitrate level. & $s_4[a]=S_{t+1}(a)/1000^2$ \\
5 & Chunks remaining & Number of chunks remaining in the video. & $s_5=\min(C_t,48)/48$ \\
\bottomrule
\end{tabularx}
\end{adjustbox}
\caption{Pensieve policy observation features used by \sys.
Pensieve represents state as a $6 \times 8$ matrix: six feature channels over the most recent eight decision steps.
The bitrate levels are 300, 750, 1200, 1850, 2850, and 4300 Kbps.
Here $B_t$ is buffer seconds, $S_t$ is downloaded chunk size in bytes, $d_t$ is download delay in ms, $S_{t+1}(a)$ is the next-chunk size at bitrate action $a$, and $C_t$ is chunks remaining.
The observation contains history for bitrate, buffer, throughput, delay, and chunks remaining, plus the current next-chunk size vector for the six bitrate choices.}
\label{tab:pensieve_state_features}
\end{table*}

\begin{table*}[t]
\centering
\footnotesize
\begin{adjustbox}{width=\textwidth,center}
\begin{tabularx}{\textwidth}{>{\raggedright\arraybackslash}p{0.20\textwidth}>{\raggedright\arraybackslash}p{0.34\textwidth}>{\raggedright\arraybackslash}X}
\toprule
Component & Feature form & Meaning \\
\midrule
Policy observation & $s_t=[L_0,L_1,\ldots,L_{N-1},J_t]$ & Park observes one load proxy per server and the incoming job size, with $N=10$ in the default setup. \\
Server load proxies & \texttt{LoadProxy0}, \ldots, \texttt{LoadProxy9} & Each $L_i$ estimates the observable load on server $i$. \\
Load-proxy value & $L_i=\sum_{j\in Q_i}\mathrm{size}(j)+\mathbf{1}[\text{server }i\text{ is running}]\cdot\max(0,\mathrm{finish\_time}_i-t)$ & The proxy combines queued job sizes with the remaining time of the currently running job. \\
Incoming job & \texttt{IncomingSize}, with $J_t=\mathrm{size}(\text{incoming job})$ & The size of the job currently waiting to be assigned. \\
Clipping & \texttt{obs\_high}, typically 500000.0 & All observation values are clipped by the observation upper bound. \\
Not in policy input & Future arrivals, oracle gap, reference actions, explicit queue lengths, server identities beyond vector position, and rule labels & These quantities are not part of the original Park policy observation. \\
Rule-derived features & Total load, mean load, minimum and maximum load, load ranks, and load-plus-incoming per server & These interpretable features are derived from the same observation for rule learning and are not extra inputs to the original Park policy. \\
\bottomrule
\end{tabularx}
\end{adjustbox}
\caption{Park policy observation and rule-derived features.
The original Park policy sees only per-server load proxies plus the incoming job size.
\sys can derive additional interpretable predicates from the same observation without changing the policy input.}
\label{tab:park_state_features}
\end{table*}

\section{Additional Evaluation Statistics}
\label{app:eval-details}

This appendix records the detailed numeric values, secondary comparisons, and per-method breakdowns omitted from Section~\ref{sec:eval}.

\subsection{Detailed Results for \ref{e1}}

For Pensieve, \sys reaches mean log performance gap 4.84, versus 4.00 for Random, 1.29 for Indago, 1.21 for Genet, 1.19 for Gilad et al., and 1.09 on Normal.
On the original scale, that is 6.95\x the Random gap and 5641\x the Normal gap, with maximum 5.14.
For Sage, \sys reaches mean relative performance gap 63.84\%, versus 40.66\% for Random and 1.66\% on Normal, with maximum 70.21\%.
For Park, \sys reaches mean relative performance gap 43.49\%, versus 7.05\% for Gilad et al. and 0.29\% on Normal, with maximum 203.80\%.

\subsection{Detailed Results for \ref{e2} and \ref{e3}}

On the most challenging family for each application, \sys removes 85.04\% of the normalized performance gap for Pensieve, 80.01\% for Sage, and 79.12\% for Park.
For transfer, the same \sys configuration removes 75.15\% on Random and 22.09\% on Indago for Pensieve, 90.08\% on Gilad et al. and 98.54\% on Genet for Sage, and 68.20\% on Gilad et al. and 36.01\% on Indago for Park.
On Normal, \sys still reduces the normalized performance gap by 18.72\% for Pensieve, 33.38\% for Sage, and 36.11\% for Park.
By contrast, retraining baselines increase the Normal gap by up to 24.25\% for Pensieve, approximately 15\x for Sage, and 81.11\% for Park.
Fig.~\ref{fig:ablation} provides the complementary ablation result.
Across all three applications, the higher-gap counterfactual sources in that figure consistently yield higher average protection percentages, while weaker sources yield much weaker protection.

\subsection{Detailed Results for \ref{e4} and \ref{e5}}

For Pensieve, mean decision-time ratio is 3.67\% without protection and 11.22\%, 32.36\%, 17.88\%, and 27.98\% across ReGuard iterations 0--3.
For Sage, the corresponding values are 41.37\% without protection and 50.11\%, 47.55\%, 49.11\%, and 48.28\% across ReGuard iterations 0--3.
For Park, the corresponding values are 2.00\% without protection and 2.41\%, 4.56\%, 6.52\%, and 1.85\% across ReGuard iterations 0--3.
Across all three systems, every iteration remains below 100\% of the available decision budget, and later iterations do not induce systematic runtime growth despite stronger protection.
For Pensieve, the mean gap drops from 4.84 at iteration 0 to 1.93, 1.22, and 1.24 across iterations 1--3, closing 96.6\% of the distance to Normal by iteration 2.
For Sage, the mean gap drops from 63.84\% to 15.28\%, 3.82\%, and 3.78\%, a 94.08\% reduction by iteration 3 and a 96.59\% reduction in distance to Normal.
For Park, the mean gap drops from 33.45\% to 6.16\%, 3.21\%, and 2.43\%, a 92.73\% reduction by iteration 3 and a 93.53\% reduction in distance to Normal.

\section{Additional Failure Analyses}
\label{app:additional-analyses}

Due to space constraints, the main paper keeps only the Park case study in Section~\ref{sec:analysis}.
This appendix contains detailed supporting statistics for Park and the corresponding Pensieve and Sage analyses.
Below, $\mathrm{p}k$ denotes the corresponding $k$th-percentile predicate computed from normal scenarios (\S\ref{sec:shield}).

\subsection{Park: Detailed supporting statistics}
\label{app:analysis-park-details}

This subsection records the detailed quantitative support omitted from Section~\ref{sec:analysis}.
In Fig.~\ref{fig:park_failures1}, the 107-decision slice consists of 69.2\% jobs of at most 0.2KB, including 45.8\% in the 0.126--0.525KB range that the fast servers can easily absorb, while only 10.3\% are 5.0KB jobs appearing in four short bursts.
For jobs of at most 0.2KB, Park chooses a slow server 97.3\% of the time and server 1 alone 86.5\% of the time, whereas Least Completion Time routes 81.1\% of those same jobs to the fast tier and 71.6\% to server 9.
The large-gap region covers 76.6\% of the slice.
During that region, at least one fast server is idle on every step and server 9 is idle on every step, yet Park still sends traffic to a slow server on 86.0\% of steps and to server 1 on 76.6\% of steps.
Inside that same region, Park never chooses a fast server, achieves mean chosen service rate 0.25 versus 0.97 for Least Completion Time, selects a zero-load server 0.0\% of the time, and incurs mean regret 6.24.

Fig.~\ref{fig:park_failures2} shows the same mechanism at workload scale.
Under the challenging scenarios discovered by \sys, Park places 76.3\% of all decisions on server 1, up from 5.6\% on normal workloads.
Its slow-server share rises from 29.1\% to 81.5\%, while its fast-server share falls from 53.1\% to 14.1\%.
Least Completion Time under the same scenarios still routes 80.6\% of decisions to the fast tier and 56.8\% to server 9.
At least one fast server remains idle on 98.0\% of steps, yet Park still chooses a slow server with a fast server idle on 81.3\% of steps and chooses server 1 while server 9 is idle on 60.8\% of steps.

The recovery rules also show that reopening the fast tier is not limited to a single server.
One additional correction rule is
\[
\begin{aligned}
&(\mathrm{Load}_9 \le \mathrm{p25}) \land (\mathrm{PostLoad}_2 \le \mathrm{p25}) \\
&{}\land (\mathrm{MeanLoad} \le \mathrm{p25}) \land (\mathrm{StdLoad} \le \mathrm{p25}) \implies \mathrm{Allow}_7=\texttt{true} .
\end{aligned}
\]
This rule says that one fast server already looks light, another candidate would still remain light after taking the job, and the overall cluster is neither heavily loaded nor badly imbalanced.
It shows that \sys is reopening the fast tier rather than hard-coding a single dispatch.

\subsection{Pensieve: Sparse-bandwidth regimes trigger bitrate overshoot}
\label{app:analysis-pensieve}

\begin{figure}[t]
    \centering
    \begin{adjustbox}{width=0.99\linewidth,center=0pt}
    \includegraphics[width=\linewidth]{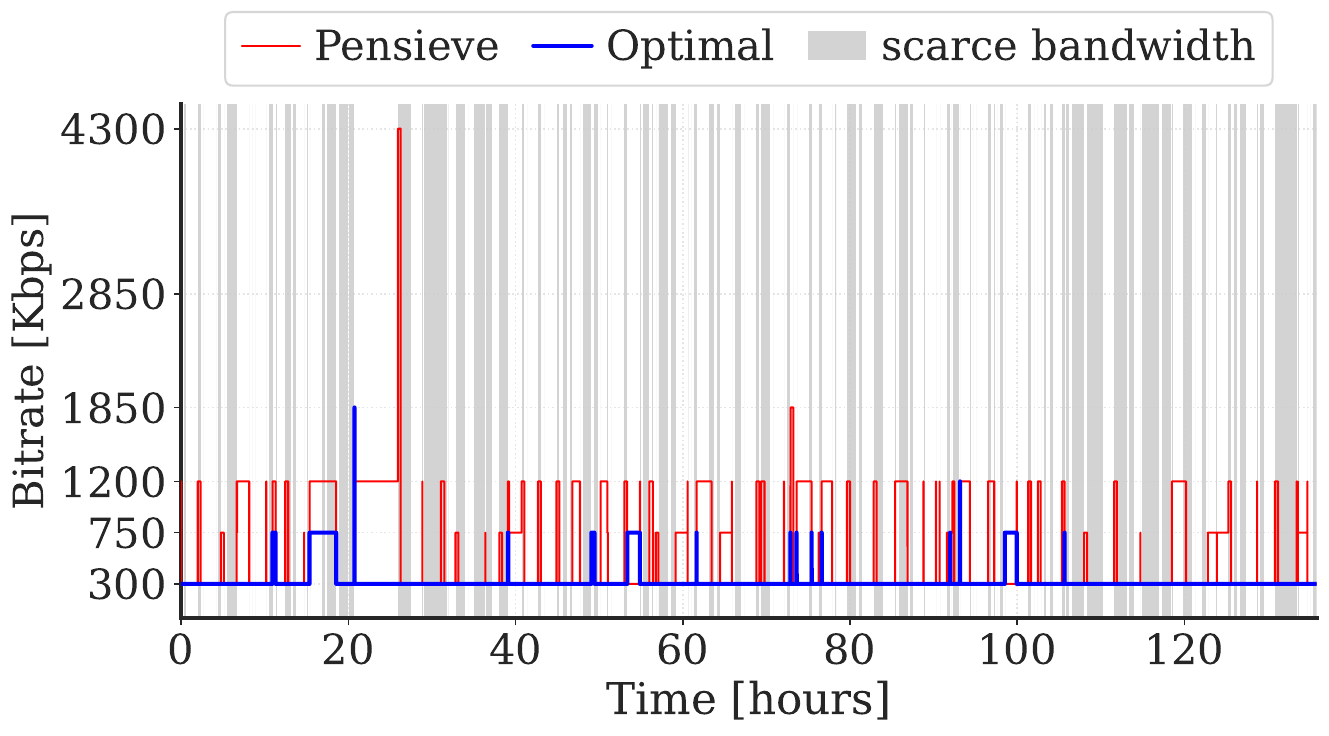}
    \end{adjustbox}
\caption{Sparse-bandwidth scenarios expose a systematic over-aggressive bitrate policy in Pensieve.
Even when the realized link remains near the bottom of the bitrate ladder, Pensieve repeatedly selects much larger rates than the oracle, creating a large avoidable performance gap.}
\label{fig:pensieve_failures}
\end{figure}

Pensieve's failure is sustained over-aggression under long sparse-bandwidth regimes rather than a few isolated bitrate mistakes.
Fig.~\ref{fig:pensieve_failures} shows scarce-bandwidth spans that cover 35.5\% of a trace duration but contain 78.6\% of all 6,248 bitrate decisions.
This concentration matters because the failure occupies a dominant operating regime instead of a few outliers.
Within that regime, Pensieve selects 465.6Kbps on average, whereas the oracle selects 304.2Kbps, so Pensieve is 53.1\% higher.
When Pensieve overshoots, the overshoot is substantial rather than marginal.
On the 21.3\% of decisions where Pensieve exceeds the oracle, it averages 1073.0Kbps versus the oracle's 308.8Kbps, with a mean gap of 764.2Kbps and a median gap of 900Kbps.
Even inside scarce-bandwidth periods alone, Pensieve still averages 422.6Kbps versus 300.9Kbps for the oracle and chooses at least 1200Kbps on 10.4\% of steps, while the oracle does so on only 0.04\%.
The largest mismatch reaches 4300Kbps for Pensieve versus 300Kbps for the oracle under realized bandwidth of only 0.001Mbps.

The mechanism behind this failure is straightforward.
Under repeated bandwidth scarcity, Pensieve keeps treating the state as compatible with aggressive upgrades even after multiple recent throughput samples collapse and chunk-download delays explode.
The result is not merely a suboptimal bitrate ranking.
It is a repeated decision to request chunks whose sizes the sparse link cannot support, which magnifies rebuffering and drives the large performance gap in Fig.~\ref{fig:pensieve_failures}.

The logic rules used by \sys make that mechanism explicit.
One representative rule is
\[
\begin{aligned}
&(\mathrm{Buffer} \le \mathrm{p10}) \land (\mathrm{Throughput}_4 \le \mathrm{p10}) \\
&{}\land (\mathrm{Delay}_3 > \mathrm{p95}) \land (\mathrm{NextChunk}_3 > \mathrm{p90}) \implies \mathrm{Bitrate}\le750 .
\end{aligned}
\]
This rule says that when the buffer is already shallow, several recent throughput samples are in the bottom decile, recent download delays are in the worst 5\%, and even the next mid-range chunk is large, Pensieve should not climb above 750Kbps.
Its semantic meaning matches the failure directly: sparse bandwidth is not a one-step dip, so aggressive upgrades only deepen the mismatch between requested chunks and available capacity.

A second rule is an even harder clamp:
\[
\begin{aligned}
&(\mathrm{Throughput}_5 \le \mathrm{p25}) \land (\mathrm{Throughput}_6 \le \mathrm{p10}) \\
&{}\land (\mathrm{Delay}_7 > \mathrm{p95}) \land (\mathrm{NextChunk}_2 > \mathrm{p10}) \implies \mathrm{Bitrate}\le300 .
\end{aligned}
\]
This rule says that when several consecutive throughput samples are near the floor, the most recent delay is extreme, and even moderate bitrate choices imply large chunks, the next action should be capped at the minimum bitrate.
That is precisely the kind of human-readable explanation the figure calls for.
The rule does not say that the network is simply bad.
It says that this particular combination of sparse recent throughput, severe delay, and large pending chunks is a high-impact state in which Pensieve's aggressive policy is likely to fail.

\subsection{Sage: Sharp drops trigger overreaction and slow reopening}
\label{app:analysis-sage}

\begin{figure}[t]
    \centering
    \begin{adjustbox}{width=0.99\linewidth,center=0pt}
    \includegraphics[width=\linewidth]{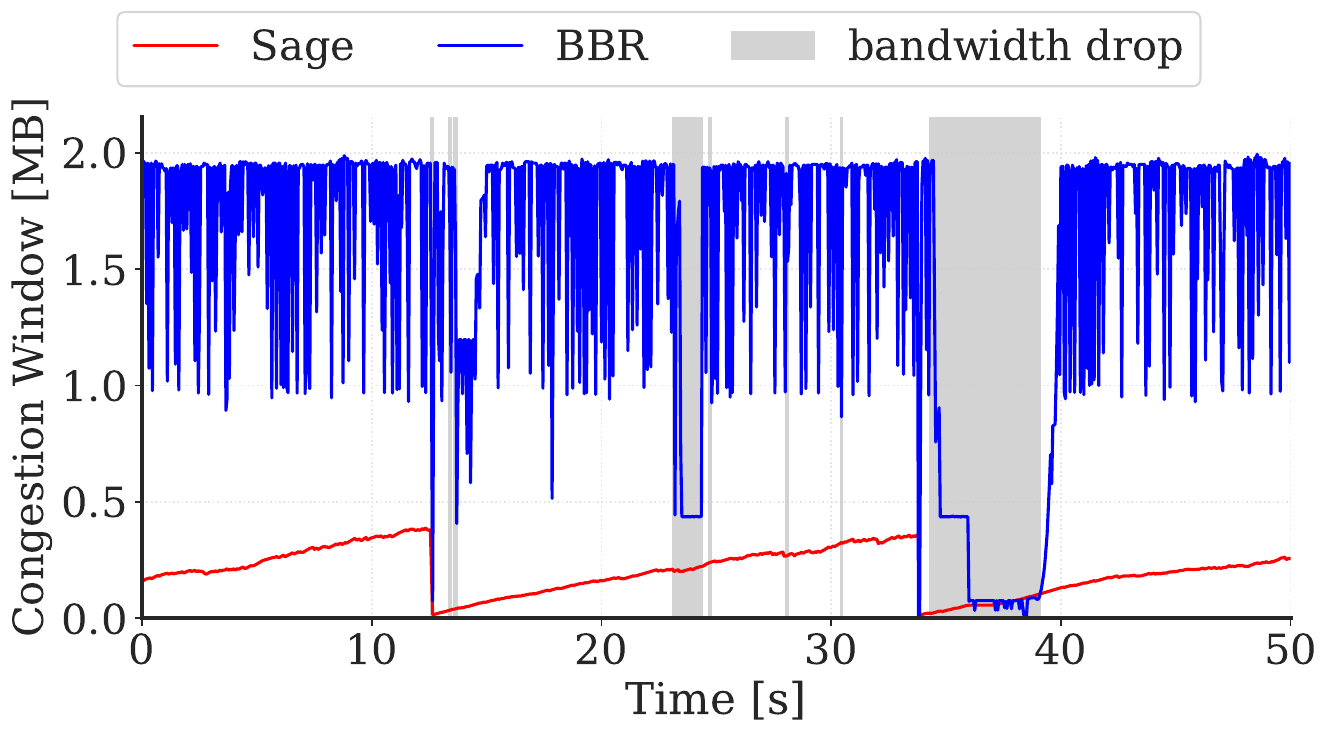}
    \end{adjustbox}
\caption{Sharp bandwidth drops expose Sage's slow recovery.
Sage collapses its congestion window far more than BBR and reopens it much more slowly, leaving throughput far below what the path can still support.}
\label{fig:sage_failures}
\end{figure}

Sage's failure is not generic weakness on a challenging path.
As Fig.~\ref{fig:sage_failures} shows, the dominant mechanism is an excessive reaction to bandwidth drops followed by very slow reopening.
Across the plotted 50s span, Sage's congestion window averages 0.20MB, or 139 packets, whereas BBR averages 1.56MB, or 1066 packets.
The medians are 0.21MB for Sage and 1.94MB for BBR, so BBR's window is 9.4\x larger at the median and larger than Sage's at 96.9\% of sampled times.
The sharpest drop, from 150Mbps to 5Mbps at 12.6s, makes Sage collapse from 261 packets to 11 packets.
Sage then reaches only 50 packets after 2.45s, 100 packets after 6.3s, and half of its pre-drop window only after 9.15s.
BBR also cuts back at that instant, but it rebounds to 827 packets within 50ms and to 1342 packets within 100ms, regaining half of its pre-drop window in just 50ms.
Over the first second after the drop, Sage averages only 19 packets while BBR averages 1097 packets, a 58\x gap.
Even after the longest 5Mbps episode ends, BBR climbs back above 1000 packets within 0.95s of bandwidth returning to 150Mbps, whereas Sage still needs 1.85s merely to exceed 100 packets.

The key weakness is therefore not the initial backoff alone.
It is that Sage continues acting as if the path remains dangerous long after the immediate drop has passed.
That persistent conservatism keeps its congestion window tiny, which in turn keeps realized throughput tiny.
Fig.~\ref{fig:sage_failures} therefore exposes a specific control failure: Sage brakes hard on sudden drops and then reopens far too slowly.

The onset of that failure appears in rules such as
\[
\begin{aligned}
&(\mathrm{LossDelta} \ge \mathrm{p95}) \land (\mathrm{LossMin} \ge \mathrm{p95}) \\
&{}\land (\mathrm{Rtt} \le \mathrm{p25}) \land (\mathrm{WindowedRate} \le \mathrm{p10}) \implies \textsc{push\_harder}.
\end{aligned}
\]
This rule says that loss has just surged, but RTT is still low while the windowed delivery rate has already fallen into the bottom decile.
That combination is a compact signature of overreaction.
The controller has already slammed on the brakes, yet the latency signal does not justify staying that conservative.

The aftermath appears in rules such as
\[
\begin{aligned}
&(\mathrm{LossMax} \le \mathrm{p25}) \land (\mathrm{LossMin} \le \mathrm{p25}) \\
&{}\land (\mathrm{MaxRate} \le \mathrm{p10}) \land (\mathrm{AvgRate} \le \mathrm{p10}) \implies \textsc{push\_harder}.
\end{aligned}
\]
This rule is important because it no longer describes an acute crisis.
Loss is now low, RTT inflation is not severe, and yet both the recent maximum and the recent average delivery rates remain near the floor.
In plain language, the network is no longer screaming ``back off,'' but Sage is still behaving as though it is.

A third rule isolates the slow-recovery phase even more directly:
\[
\begin{aligned}
&(\mathrm{CurrentRate} \le \mathrm{p10}) \land (\mathrm{RateGrowth} \le \mathrm{p10}) \\
&{}\land (\mathrm{AvgRate} \le \mathrm{p10}) \land (\mathrm{RateVsMax} \le \mathrm{p25}) \implies \textsc{push\_harder}.
\end{aligned}
\]
This rule says that Sage is not only small in absolute terms.
It is growing slowly and operating far below what it was recently capable of.
That is exactly the behavior seen in Fig.~\ref{fig:sage_failures}, and it makes the failure interpretable in a way that the raw congestion-window trace alone cannot.

\section{Testbed}
\label{apdx:testbed}
For fair evaluation, we conduct all experiments on a two-socket server with 40 logical Intel Xeon E5-2660v3 CPUs running at 2.60GHz and 256GiB of DRAM.

\end{document}